\documentclass[10pt, twoside]{article}

\usepackage{amsmath,amsfonts,amsbsy,amssymb,xparse,mathtools,amsthm,tabulary,graphicx,times,fancyhdr, mathrsfs, pgf, epsfig, epstopdf, lscape, subfig,booktabs, enumitem,lipsum,tkz-graph,scalerel,stackengine,setspace,kbordermatrix,adjustbox,float,caption}

\usepackage{calc}
\usepackage{makecell}
\usetikzlibrary{positioning,shapes,calc,shapes.geometric,decorations.pathmorphing,arrows}
\usetikzlibrary{matrix}

\usepackage{xcolor}
\PassOptionsToPackage{usenames,dvipsnames,svgnames,table}{xcolor}

\def\titlerunning#1{\gdef\titrun{#1}}
\makeatletter
\def\author#1{\gdef\autrun{\def\and{\unskip, }#1}\gdef\@author{#1}}
\def\address#1{{\def\and{\\\hspace*{18pt}}\renewcommand{\thefootnote}{}%
\footnote {#1}}%
\markboth{\autrun}{\titrun}}
\makeatother
\def\email#1{e-mail: #1}

\def\keywords#1{\par\medskip
\noindent\textbf{Keywords.} #1}

\numberwithin{equation}{section}

\usepackage{lipsum}  
\usepackage{setspace}
\newlength{\depthofsumsign}
\GraphInit[vstyle = Shade]
\tikzset{
	LabelStyle/.style = { rectangle, rounded corners, draw,
		minimum width = 2em, fill = yellow!50,
		text = red, font = \Large\bfseries },
	VertexStyle/.append style = { inner sep=5pt,
		font = \Large\bfseries},
	EdgeStyle/.append style = {->, bend left} }

\tikzset{
	mynode/.style={
		draw,
		thick,
		anchor=south west,
		minimum width=2cm,
		minimum height=1.3cm,
		align=center, 
		inner sep=0.2cm,
		outer sep=0,
		rectangle split, 
		rectangle split parts=2,
		rectangle split draw splits=false},
	reverseclip/.style={
		insert path={(current page.north east) --
			(current page.south east) --
			(current page.south west) --
			(current page.north west) --
			(current page.north east)}
	}
}

\tikzstyle{io} = [trapezium, trapezium left angle=70, trapezium right angle=110, minimum width=3cm, minimum height=0.5cm, text centered, draw=black, fill=blue!30]
\tikzstyle{process} = [rectangle, minimum width=3cm, minimum height=1cm, text centered, draw=black, fill=orange!30]
\tikzstyle{process2} = [rectangle, minimum width=3cm, minimum height=1cm, text centered, draw=black, fill=green!30]
\tikzstyle{process3} = [rectangle, minimum width=2cm, minimum height=1cm, text centered, draw=black, fill=yellow!30]

\newtheorem{axiom}{}

\newtheorem{theorem}{Theorem}
\newtheorem{lemma}{Lemma}
\newtheorem{corollary}{Corollary}

\newtheorem{definition}{Definition}
\newtheorem{example}{Example}

\newtheorem{remark}{Remark}

\newcommand{\CommaPunct}{\mathpunct{\raisebox{0.5ex}{,}}}

\newcommand{\blue}[1]{\textcolor{blue}{#1}}

\DeclareFontFamily{U}{mathx}{\hyphenchar\font45}
\DeclareFontShape{U}{mathx}{m}{n}{
	<-6> mathx5 <6-7> mathx6 <7-8> matha7
	<8-9> mathx8 <9-10> mathx9
	<10-12> mathx10 <12-> mathx12
}{}
\DeclareSymbolFont{mathx}{U}{mathx}{m}{n}
\DeclareMathSymbol{\bigominus}{\mathop}{mathx}{"C1}

\DeclareMathAlphabet\mathbfcal{OMS}{cmsy}{b}{n}

\stackMath
\newcommand\reallywidehat[1]{%
	\savestack{\tmpbox}{\stretchto{%
			\scaleto{%
				\scalerel*[\widthof{\ensuremath{#1}}]{\kern-.6pt\bigwedge\kern-.6pt}%
				{\rule[-\textheight/2]{1ex}{\textheight}}%WIDTH-LIMITED BIG WEDGE
			}{\textheight}% 
		}{0.5ex}}%
	\stackon[1pt]{#1}{\tmpbox}%
}

\usepackage[backend=biber, style=authoryear, autocite=inline, sorting = nyt, sortcites = true, maxbibnames=99, maxcitenames=4,hyperref=true, natbib=true]{biblatex}
\usepackage[colorlinks=true,citecolor=blue,urlcolor=blue,bookmarks=true,hypertexnames=true,linkcolor=blue,
hypertexnames=true]{hyperref}
\usepackage[noabbrev,nameinlink]{cleveref}

\makeatletter
\newcommand*{\transpose}{%
	{\mathpalette\@transpose{}}%
}
\newcommand*{\@transpose}[2]{%
	\raisebox{\depth}{$\m@th#1\intercal$}%
}
\makeatother

\ExplSyntaxOn

\NewDocumentCommand \rvect { s o m }
{
	\IfBooleanTF {#1}
	{ \vectaux*{#3} }
	{ \IfValueTF {#2} { \vectaux[#2]{#3} } { \vectaux{#3} } }
	^{\transpose}
}

\DeclarePairedDelimiterX \vectaux [1] {\lbrack} {\rbrack}
{ \, \dbacc_vect:n { #1 } \, }

\cs_new_protected:Npn \dbacc_vect:n #1
{
	\seq_set_split:Nnn \l_tmpa_seq { , } { #1 }
	\seq_use:Nn \l_tmpa_seq { \enspace }
}
\ExplSyntaxOff

\DeclarePairedDelimiterX\set[1]\lbrace\rbrace{#1}

\newcommand*\oline[1]{%
	\vbox{%
		\hrule height 0.5pt%                  % Line above with certain width
		\kern0.25ex%                          % Distance between line and content
		\hbox{%
			\kern-0.1em%                        % Distance between content and left side of box, negative values for lines shorter than content
			\ifmmode#1\else\ensuremath{#1}\fi%  % The content, typeset in dependence of mode
			\kern0.0em%                        % Distance between content and left side of box, negative values for lines shorter than content
		}% end of hbox
	}% end of vbox
}

\begin{document}
\sloppy
%%%%% To ease editing, add:

\baselineskip=22pt

%%%%%%%%%%%%%%%%

\titlerunning{Theory and Applications of Financial Chaos Index}

\title{\Large Theory and Applications of Financial Chaos Index}

\author{Masoud Ataei
\and 
Shengyuan Chen
\and 
Zijiang Yang
\and 
M.Reza Peyghami}

\date{}

\maketitle

\address{M. Ataei: Department of Mathematics and Statistics, York University, Ontario, Canada;\\ \email{mataei@yorku.ca}
\and
S. Chen: Department of Mathematics and Statistics, York University, Ontario, Canada;\\ \email{chensy@mathstat.yorku.ca} 
\and 
Z. Yang: School of Information Technology, York University, Ontario, Canada;\\ \email{zyang@yorku.ca}
\and 
M.Reza Peyghami: Department of Mathematics and Statistics, York University, Ontario, Canada;\\ \email{rpeygham@yorku.ca}  }

%\subjclass{Primary 62A01; Secondary 47H05}

%%%%%%%%%%%%%%%%%%%%%%%%%%%%%%%%%%%%%%%%%%%%%%%%%%%%%%%%%%%%%%%%%%%%%%%%%%%%%%%%%%%%%%%%%%%%%%%%%%%%%%%%%%%%%%%%%%%%%%%%

\begin{abstract}
We develop a new stock market index that captures the chaos existing in the market by measuring the mutual changes of asset prices. This new index relies on a tensor-based embedding of the stock market information, which in turn frees it from the restrictive value- or capitalization-weighting assumptions that commonly underlie other various popular indexes. We show that our index is a robust estimator of the market volatility which enables us to characterize the market by performing the task of segmentation with a high degree of reliability. In addition, we analyze the dynamics and kinematics of the realized market volatility as compared to the implied volatility by introducing a time-dependent dynamical system model. Our computational results which pertain to the time period from January 1990 to December 2019 imply that there exist a bidirectional causal relation between the processes underlying the realized and implied volatility of the stock market within the given time period, where it is shown that the later has a stronger causal effect on the former as compared to the opposite. This result connotes that the implied volatility of the market plays a key role in characterization of the market's realized volatility.

\end{abstract}

%%%%%%%%%%%%%%%%%%%%%%%%%%%%%%%%%%%%%%%%%%%%%%%%%%%%%%%%%%%%%%%%%%%%%%%%%%%%%%%%%%%%%%%%%%%%%%%%%%%%%%%%%%%%%%%%%%%%%%%%

%%%%%%%%%%%%%%%%%%%%%%%%%%%%%%%%%%%%%%%%%%%%%%%%%%%%%%%%%%%%%%%%%%%%%%%%%%%%%%%%%%%%%%%%%%%%%%%%%
%%%%%%%%%%%%%%%%%%%%%%%%%%%%%%%%%%%%%%%%%%%%%%%%%%%%%%%%%%%%%%%%%%%%%%%%%%%%%%%%%%%%%%%%%%%%%%%%%
%%%%%%%%%%%%%%%%%%%%%%%%%%%%%%%%%%%%%%%%%%%%%%%%%%%%%%%%%%%%%%%%%%%%%%%%%%%%%%%%%%%%%%%%%%%%%%%%% 

\section{Introduction}
\label{Section:1}

In this paper, we set out to analyze the stock markets by resorting to the theory and applications of tensors. We develop a special class of spatio-temporal tensors and show that our created algebraic object is a suitable structure to embed the collective judgment of agents who are present in the market. It is then the use of such effective embedding of the stock market information that allows us to take advantage of different mathematical and statistical properties of the constructed tensors, in order to define an index which captures the chaos existing in the market by measuring the mutual changes of asset prices. The proposed financial chaos index is shown to be indeed a robust estimator of market volatility, which in turn enables us to reliably utilize it for various purposes including analysis of market regimes, performing the task of market segmentation and testing the efficient market hypothesis.

To begin with, we formulate our tensor-based model of the stock market on the basis of a rich body of literature that has been conducted by numerous scholars during the past decades who have developed the theory of pairwise comparative judgments to its near maturity stage. This theory thus far has mainly relied on mathematical aspects of the so-called pairwise comparison matrices and their special subclass widely known as the reciprocal pairwise comparison matrices, which finds various application in different areas of science including the analytic hierarchy processes, e.g., see \Citep{ishizaka2009analytic, ishizaka2011review, brunelli2014introduction, saaty2012models, mu2017understanding, mu2016practical, mu2017practical} for review of the foundations and recent developments. Our primary contribution in this work is to extend the theory of pairwise comparative judgments to a tensor domain, using the invaluable findings available on the pairwise comparison matrices and their various properties.

To achieve the above-mentioned goal, we first axiomitize the pairwise comparative judgments, and in light of the defined axiom, we state and re-interpret and expand on the most important results pertaining to mathematical and statistical properties of the pairwise comparison matrices and their special subclass of reciprocal matrices. By examining the various properties of the mentioned comparison matrices, we demonstrate that these mathematical objects are among the most suitable structures, if not the only ones, to  model the procedures that underlie agents' thought processes, especially those that relate to their judgment mechanisms. We then extend the comparison matrices to higher dimensions and establish an important result on relation of the polyadic decomposition of our constructed tensors and the consistency of the judgments cast by the agents throughout time. Once having constructed the mentioned tensors based upon the pairwise comparison matrices and investigated their various algebraic properties, we utilize them to lay out the foundations for defining the financial chaos index.

Being a robust estimator of the market's realized volatility, subsequently raises an important question on the relation of our proposed financial chaos index as compared to other indexes that measure rather the market's implied volatility. Our further contribution in this work is to examine the relationship between the financial chaos index and VIX, which is a predominant index for measuring the market's implied volatility. Investigating the relationship between the financial chaos index and VIX, in turn enables us to characterize the relation between the equity and option markets. We systematically study the kinematics and dynamics of the financial chaos index and VIX using a variety of tools. Namely, we characterize the relationship between these two indexes by carrying out their fractional cointegration analysis, which in turn leads to identification of their long-run equilibrium relationship. Furthermore, by using the orthogonalized impulse-response functions as well as information-theoretic measures, we determine the possible causal relations existing between two indexes. Also, we formulate a time-dependent dynamical system to model the flow of various types of information such as self-entropy, transfer entropy and approximate entropy between the considered indexes over time.

\subsection{Outline}
\label{Section:1_1}

The remainder of this paper and its structure are organized as follows: In \Cref{Section:2}, we define the axiom of pairwise comparative judgment, under which various algebraic properties of the pairwise comparison matrices, and in particular reciprocal pairwise comparison matrices, are elucidated. In this section, we further define a third-order spatio-temporal tensor object using the so-called constitution criterion which embeds the stock market information, which is then followed by the definition of the so-called consensus tensor. Further, in \Cref{Section:3} various tools are employed to characterize the relationship between our proposed index and the market's implied volatility index (VIX).

\subsection{Preliminaries}
\label{Section:1_2}
Throughout the paper, scalar quantities, vectors and matrices are denoted using lowercase lightface (e.g., $x$), lowercase boldface (e.g., $\mathbf{x}$) and uppercase boldface (e.g., $\mathbf{X}$) letters, respectively, where all vectors are presumed to be column ones. Also, boldface calligraphic letters (e.g., $\mathbfcal{X}$) are used to denote the tensors. We further use $\mathbb{R}_{\geq 0}$ and $\mathbb{R}_{>0}$ to denote the fields of positive and strictly positive real numbers, respectively.

Tensors are algebraic objects which extend the notion of matrices to higher dimensions. The number of dimensions of a tensor is referred to as the \textit{order} or \textit{modes} of the tensor, i.e., $\mathbfcal{X}\in~\mathbb{R}^{I_1\times I_2\dots\times I_N}$ indicates an $N$th-order tensor having $K=\prod_{n=1}^{N}I_n$ elements in total. Analogous to matrix rows and columns, \textit{mode-n fibers} of tensors are derived by fixing every index but one, while fixing all but two indexes of the tensors yields hyperplanes known as \textit{slices}. The Frobenius norm of $\mathbfcal{X}$ is further defined as follows:
\begin{equation}
\left\Vert \mathbfcal{X} \right\Vert_F = \sqrt{\langle \mathbfcal{X},\mathbfcal{X} \rangle} = \sum\limits_{i_1=1}^{I_1} \sum\limits_{i_2=1}^{I_2} \dots \sum\limits_{i_n=1}^{I_N} \sqrt{x_{i_1 i_2 \dots i_n}^2} \,.
\end{equation}

An important technique for tensor decomposition which finds a wide range of applications in science is the \textit{polyadic decomposition} which approximates $\mathbfcal{X}$ by the sum of rank-$1$ tensors (also called \textit{atoms}), i.e.,
\begin{equation}
\label{Eq:CP}
\mathbfcal{X} \approx \widehat{\mathbfcal{X}} = \sum\limits_{r=1}^{R} \mathbf{a}_r \circ \mathbf{b}_r \circ \mathbf{c}_r,
\end{equation}
where the notation "$\circ$" represents \textit{vector outer product}, and $R$ is a given positive integer known as the rank of $\widehat{\mathbfcal{X}}$. It is worth mentioning that in the case where $R$ is identified to be a minimal rank, the decomposition given by \cref{Eq:CP} is referred to as the \textit{canonical polyadic decomposition} (CPD). Furthermore, the tensor $\mathbfcal{X}$ is said to be a rank-$1$ tensor if it can be written as the outer product of three vectors, i.e.,
\begin{equation}
\mathbfcal{X}  =  \mathbf{a} \circ \mathbf{b} \circ \mathbf{c},
\end{equation}
where $\mathbf{a}\in \mathbb{R}^{I_1}$, $\mathbf{b}\in \mathbb{R}^{I_2}$ and $\mathbf{c}\in \mathbb{R}^{I_3}$.

%%%------------------------------------------------------------------------------------------------------
%%%------------------------------------------------------------------------------------------------------
%%%------------------------------------------------------------------------------------------------------
%%%------------------------------------------------------------------------------------------------------

\section{Phenomenology of Comparative judgments}
\label{Section:2}

We consider a general exchange economy consisting of a universe $\mathit{S}:=\{S_1,S_2,\dots,S_N\}$ of $N$ assets traded by $M\gg 1$ agents over a finite time horizon $ \mathcal{T} := \{t_1,t_2,\dots\} \subset \{1,2,\dots,T\}$, where $T$ denotes the maximum number of discrete time periods under consideration. Given the set of alternatives $\mathit{S}$, we are then concerned with the situation where an omniscience agent, say \textit{agent $m$}, whose judgment is assumed to typify the collective judgment of the market participants, performs a complete set of pairwise comparisons on alternatives contained in $\mathit{S}$ and assigns a strictly positive quantity $\mu_m^{(t)}$ to her \textit{degree of preference} on every comparison at time $t\in~\mathcal{T}$. It will be shown that under mild conditions of the so-called \textit{pairwise comparative judgment} axiom and the strict positivity of $\mu_m^{(t)}$ for every $t\in~\mathcal{T}$, the preference gauge assigned by the agent $m$ to every comparison is bound to take on form of a ratio between two strictly positive quantities. Various consequences of such an implication and its phenomenological aspects that underlie the faculty of judgment in human beings will be also discussed throughout this section.

\subsection{Pairwise Comparison Matrices}
\label{Section:2_1}
Let us assume that for the agent $m$, there exist a directed and complete edge-weighted \textit{comparison multigraph} $\mathcal{G}_m^{(t)}=(\mathit{S},\mathcal{E}_m^{(t)})$ with $N$ nodes, whose edge set $\mathcal{E}_m^{(t)}$ contains the preference degrees assigned by the agent to every pair of the assets that are compared at time $t\in \mathcal{T}$. A \textit{walk} or \textit{path} on $\mathcal{G}_m^{(t)}$ is defined as a sequence $S_1,S_2,\dots,S_k$ of nodes where $(S_i,S_{i+1})\in \mathcal{E}_m^{(t)}$ for $i=1,2,\dots,k-1$, whose length is determined by the number of edges traversed to reach the terminal node starting from the initial node. In this setting, we presume that traversing a self-loop would increase the walk length by an increment. By a cycle we understand a closed walk with identical initial and terminal nodes. The weight matrix associated to $\mathcal{G}_m^{(t)}$ is defined as follows:
\begin{definition}[PCM]
	\label{Def:PCM}
	A square matrix $\mathbf{A} \in \mathcal{A}_{\mathrm{PCM}}^{(t)} \subset \mathbb{R}_{>0}^{N\times N}$ is said to be a \textit{pairwise comparison matrix} (PCM) if its elements elicit the agent's preference degrees on pairs of assets that are compared at a given time $t \in \mathcal{T}$, where $\mathcal{A}_{\mathrm{PCM}}^{(t)}$ represents the family of all pairwise comparison matrices at time $t$.
\end{definition}

According to this definition, the $(i,j)$-th element of the $k$-th power $\mathbf{A}^k$ of the pairwise comparison matrix  $\mathbf{A}$ can be interpreted as the preference degree of asset $S_j$ over $S_i$ based on a walk of length $k$ performed on the comparison multigraph $\mathcal{G}_m^{(t)}$. Then completeness of $\mathcal{G}_m^{(t)}$ guarantees that the weight matrix $ \mathbf{A}$ is irreducible as every pair of nodes are connected by a path of arbitrary a certain length. Moreover, since every pair of the nodes in $\mathcal{G}_m^{(t)}$ is connected via a certain path of length $k$, the matrix is primitive, and such a primitivity property strengthens the irreducibility condition by further positing $k$-connectivity of~$\mathcal{G}_m^{(t)}$.

Let us now introduce the relations $\prec_k$, $\sim_k$ and $\preceq_k$ which denote $k$-\textit{length} \textit{strict preference}, \textit{indifference} and \textit{preference-indifference} relations on $\mathit{S}$, respectively. These relations, which extend those proposed in \Citep{suppes1999introduction,fishburn1970intransitive} for the special case $k=1$, which are defined as follows:
\begin{equation}
\label{Eq:preference}
\forall \mu_m^{(t)} \in R_{>0}, \exists \prec_k, \exists \sim_k: (S_i \prec_k  S_j) \lor (S_j \prec_k  S_i) \lor (S_i \sim_k  S_j).
\end{equation}
The formula \eqref{Eq:preference} can be perceived as the fact that the agent $m$ believes that for all her subjective preference degrees measured at a given time $t \in \mathcal{T}$, there exist relations $\prec_k$ and $\sim_k$ such that the asset $S_j$ would be either strictly better than the asset $S_i$ based on a walk of length $k$ on $\mathcal{G}_m^{(t)}$ (e.g., $S_i \prec_k  S_j$) or strictly worse (e.g., $S_j \prec_k  S_i$), or $S_i$ and $S_j$ should be judged as being indifferent from one another (e.g., $S_i \sim_k  S_j$). We assume that $\prec_k$ is a strict partial order and define the $k$-length indifference $\sim_k$ as the absence of strict preference such that
\begin{equation}
x\sim_k y \Leftrightarrow (x\not\prec_k y)  \wedge  (y\not\prec_k x).
\end{equation}

We are then ready to formulate the so-called pairwise comparative judgment axiom as follows:
\begin{axiom}[PCJ]
	\label{PCJ_Axiom}
	The agent $m$ is said to satisfy a \textit{pairwise comparative judgment} (PCJ) axiom, if for every $t \in \mathcal{T}$, $k\geq 1$, $i,j,l\in \{1,2,\dots,N\}$ and $\mu_m^{(t)}>0$, there exist relations $\prec_k$, $\preceq_k$ and $\sim_k$ such that
	\begin{enumerate}[label=(\roman*).]
		\item $S_i\prec_k S_j \Longleftrightarrow   \mu_m^{(t)}(S_i\prec_k S_j) > \mu_m^{(t)}(S_j\prec_k S_i);$
		\item $S_i\sim_k S_j \Longleftrightarrow   \mu_m^{(t)}(S_i\sim_k S_j) = \mu_m^{(t)}(S_j\sim_k S_i);$
		\item $\mu_m^{(t)}(S_i\preceq_k S_j) = \mu_m^{(t)}(S_i\preceq_k S_l) \,  \mu_m^{(t)}(S_l\preceq_k S_j).$
	\end{enumerate}
\end{axiom}
The properties ($i$) and ($ii$) stated in the \blue{PCJ} \ref{PCJ_Axiom} imply that for every $x,y,z\in \mathit{S}$, the $k$-length preference $\prec_k$ and indifference $\sim_k$ relations satisfy the reflexivity of indifferences, irreflexivity of preferences, symmetry of indifferences, asymmetry of preferences, transitivity of preferences and transitivity of indifferences. It is also noted that the $k$-length preference-indifference relation $\preceq_k$ is defined as the union of $k$-length preference and indifference relations, i.e.,
\begin{equation}
x\preceq_k y \Leftrightarrow (x\sim_k y)  \lor  (x\prec_k y).
\end{equation}
As a result, transitivity of $\sim_k$ amounts to the fact that $\preceq_k$ is transitive and complete, which in turn implies that $\preceq_k$ is a weak order. We point out that in the special case $k=1$, part $(iii)$ of the \blue{PCJ} \ref{PCJ_Axiom} is often referred to as the \textit{consistency property} of PCMs. It is used for quantifying the judgments consistently and ensures interconnectivity among all activities which in turn permits their quantitative assessment. Let us propose the following statistic for measuring the average degree of consistency on $\mathcal{G}_m^{(t)}$:
\begin{equation}
\label{Eq:cDeg}
\mathrm{cDeg}(\mathcal{E}_m^{(t)}) = \frac{1}{N^3}  \sum\limits_{i,j,l=1}^N  \mathbb{I}[\epsilon_{il}^{(t)}\epsilon_{lj}^{(t)}= \epsilon_{ij}^{(t)}] .
\end{equation}
Here, one has $\mathrm{cDeg}=1$ for fully consistent $\mathcal{G}_m^{(t)}$, while for $\mathrm{cDeg}(\mathcal{G}_m^{(t)})$ values reasonably close to $1$ we would rather refer to $\mathcal{G}_m^{(t)}$ as a \textit{near-consistent} comparison multigraph. Furthermore, were it to be $\mathrm{cDeg}=0$, the comparison multigraph $\mathcal{G}_m^{(t)}$ would be \textit{inconsistent}, rendering impossibility of existence of a quantitative scheme which would allow any associative comparison performed among $N^3$ triads of assets.

%%%%%%%%%%%%%%%%%%%%%%%%%%%%%%%%%%%%%%%%%%%%%%%%%%%%%%%%%%%%%%%%%%%%%%%%%%%%%%%%%%%%%%%%%%%%%%%%%%%%%%%%%%%

In light of the \blue{PCJ} \ref{PCJ_Axiom}, we present the following series of technical lemmas and theorems on properties of PCMs and their special sub-classes. 

\begin{lemma}[Rank Property]
	\label{Thrm:Rank}
	Under the \blue{PCJ} \ref{PCJ_Axiom},  it holds for the agent $m$ that for every $t\in \mathcal{T}$,
	\begin{equation}
	\forall \mathbf{A} \in \mathcal{A}_{\mathrm{PCM}}^{(t)}:  \mathrm{rank}\left(\mathbf{A}\right) =1 .
	\end{equation}
\end{lemma}
\begin{proof}
	First, note that by consistency property $(iii)$ of the \blue{PCJ} \ref{PCJ_Axiom}, the matrix $\mathbf{A}$ can be written as follows:
	\begin{equation*}
	\mathbf{A} = 
	\begin{pmatrix}
	a_{11} & a_{12} & \cdots & a_{1N} \\
	a_{21} & a_{22} & \cdots & a_{2N} \\
	\vdots  & \vdots  & \ddots & \vdots  \\
	a_{N1} & a_{N2} & \cdots & a_{NN} 
	\end{pmatrix}
	=
	\begin{pmatrix}
	a_{11} & a_{12} & \cdots & a_{1N} \\
	a_{21}a_{11} & a_{21}a_{12} & \cdots & a_{21}a_{1N} \\
	\vdots  & \vdots  & \ddots & \vdots  \\
	a_{N1}a_{11} & a_{N1}a_{12} & \cdots & a_{N1}a_{1N} 
	\end{pmatrix}  .
	\end{equation*}
	Subsequently, observe that each row of this matrix is a constant multiplier of first row. Hence, $\mathrm{rank}(\mathbf{A})=1$.		
\end{proof}

Now, let us provide some intuition behind \Cref{Thrm:Rank}. For a strictly positive matrix (which is the case for any PCM), the rank of the matrix indicates the number of vertices on the convex hull of the data contained in the matrix. Thus, $\mathrm{rank}\left(\mathbf{A}\right) =1$ delineates that a single class is sufficient to cluster the contained data, further implying existence of strong dependencies among various entries of the matrix. In other words, the fact that $\mathrm{rank}\left(\mathbf{A}\right) =1$, implies that there exist only one piece of information embedded within the matrix $\mathbf{A}$.

%%%%%%%%%%%%%%%%%%%%%%%%%%%%%%%%%%%%%%%%%%%%%%%%%%%%%%%%%%%%%%%%%%%%%%%%%%%%%%%%%%%%%%%%%%%%%%%%%%%%%%%%%%%

In turn, \Cref{Thrm:Rank} yields
\begin{lemma}[Reciprocity Property]
	\label{Thrm:Ratio_Principle}
	Under the \blue{PCJ} \ref{PCJ_Axiom}, for every $ \mathbf{A} \in  \mathcal{A}_{\mathrm{PCM}}^{(t)} $ constructed by the agent $m$, we have that for each $t\in \mathcal{T}$,
	\begin{equation}
	\label{Eq:Reciprocal_Propery}
	\begin{split}
	\forall i = 1,2,\dots,N:& \,  a_{ii} = 1,  \\
	\forall i,j = 1,2,\dots,N:& \, a_{ij} = \frac{1}{a_{ji}}\cdot
	\end{split}
	\end{equation}
	
\end{lemma}
\begin{proof}
	
	By \Cref{Thrm:Rank}, $\mathrm{rank}(\mathbf{A}) = 1$. Thus, the matrix $\mathbf{A}$ can be written as the outer product of two vectors, i.e.,
	$
	\mathbf{A} = \mathbf{u} \circ \mathbf{v}^\intercal ,
	$
	where $\mathbf{u}, \mathbf{v} \in \mathbb{R}^N_{>0}$. Subsequently, using the consistency property $(iii)$ of the \blue{PCJ} \ref{PCJ_Axiom}, for every $i,k=1,2,\dots,N$, we get
	\begin{gather*}
	a_{ii} = u_i v_i = (u_i v_k) (v_k v_i) = (u_i v_i) (u_k v_k)  \implies a_{ii} = a_{ii} (u_k v_k)  \implies u_k v_k = 1. 
	\end{gather*}
	It is then evident that
	\begin{equation}
	\label{Eq:rec}
	u_k = \frac{1}{v_k}\CommaPunct \quad \forall  k\in\{1,2,\dots,N\}.
	\end{equation}
	In view of \cref{Eq:rec}, the elements of $\mathbf{A}$ are essentially expressed by element-wise division of vectors $\mathbf{u}$ and $\mathbf{v}$, i.e.,
	$
	a_{ij} = {u_i}/{v_j}.
	$
	This then leads to
	$
	a_{ij} = {1}/{a_{ji}}
	$
	which completes the proof.	
\end{proof}

The assertions of \Cref{Thrm:Ratio_Principle} suggest us to concentrate on a particular subclass of $\mathcal{A}_{\mathrm{PCM}}^{(t)}$ which satisfy the reciprocity property at every $t\in \mathcal{T}$. It will be shown that \cref{Eq:Reciprocal_Propery} fulfilled for members of such subclass of PCMs make them a preferred choice for modeling the procedures that underlie agents' judgment faculties. In other words, employing the above-mentioned subclass of PCMs enables us to connect mechanisms of the agents' ways of thought to their quantification of the so-called \textit{secondary qualities}. In this respect, it is worth to provide the following definition.
\begin{definition}[RPCM]
	\label{Def:RPCM}
	A square matrix $\mathbf{A} \in \mathcal{A}_{\mathrm{RPCM}}^{(t)} \subset \mathcal{A}_{\mathrm{PCM}}^{(t)}$ is said to be a \textit{reciprocal pairwise comparison matrix} (RPCM) if
	\begin{enumerate}[label=(\roman*), itemsep=0pt, topsep=0pt]
		\item $\mathbf{A}$ has unit diagonal elements, i.e., $a_{ii} = 1 $ for all $i= 1,2,\dots,N$,
		\item $\mathbf{A}$ has reciprocal off-diagonal elements, i.e., $a_{ij} = \frac{1}{a_{ji}} $ for all $i,j= 1,2,\dots,N$,
	\end{enumerate}
	where $\mathcal{A}_{\mathrm{RPCM}}^{(t)}$ represents the family of all reciprocal pairwise comparison matrices at time $t$ which is a proper subclass of $\mathcal{A}_{\mathrm{PCM}}^{(t)}$.
\end{definition}

An important, plausibly unique, characterization property pertaining to RPCMs is stated in the following lemma.
\begin{lemma}[Scaled Self-similarity Property]
	\label{Thrm:power}
	Under the \blue{PCJ} \ref{PCJ_Axiom} and for the agent $m$, the following generalized idempotent relation
	\begin{equation}
	\forall k\geq 1 , \forall \mathbf{A} \in \mathcal{A}_{\mathrm{RPCM}}^{(t)}:  \mathbf{A}^k = N^{k-1} \mathbf{A}   
	\end{equation}
	holds for every $t\in \mathcal{T}$. 
\end{lemma}
\begin{proof}
	It is straightforward.		
\end{proof}

We clarify possible implications of \Cref{Thrm:power} as follows. For any RPCM associated to agent $m$'s set of beliefs, \Cref{Thrm:power} asserts that the diagonal elements of its $k$-th power are equal to those of the original RPCM (whose diagonal elements are all ones) times $N^{k-1}$. Thus, we may write
\begin{equation}
\label{Eq:Slef_Idempotent}
\forall k \geq 1, \forall i \in \{1,2,\dots,N\} , \mu_m^{(t)}(S_i\sim_k S_i) = N^{k-1} \, \mu_m^{(t)}(S_i\sim_1 S_i)   .
\end{equation}
To gain more insight on implications of the scaled self-similarity property of RPCMs and comprehend the rationale behind the results provided by \cref{Eq:Slef_Idempotent}, it deserves to distinguish between $S_i\sim_k S_i$ and $S_i\sim_1 S_i$ from the information-theoretic point of view. To this end, we first note that $S_i\sim_1 S_i$ corresponds to traversing a self-loop on $\mathcal{G}_m^{(t)}$, i.e., comparing an asset to itself. However, the information gain for such a comparison is clearly zero. Hence, we may refer to the event of performing the comparison $S_i\sim_1 S_i$ as a \textit{redundant action}. On the other hand (and in contrast to traversing self-loops), communicating the cycles of $\mathcal{G}_m^{(t)}$, i.e., $S_i\sim_k S_i$ where $k\geq 2$, can potentially increase the amount of information gain for the agent. This stems from the fact that the subjective gauge of the agent about dominance relations existing among assets becomes more refined each time she assimilates the information pertaining to the cycles of $\mathcal{G}_m^{(t)}$ whose lengths are strictly greater than one.

Now, since the relation $S_i \sim_k S_i$ corresponds to a cycle of length $k-1$ for every $i\in \{1,2,\dots,N\}$, there exist at most $k-1$ intermediary nodes in each cycle.  According to the \textit{maximum entropy principle}, we can assume that the agent's assimilation of information on each intermediary node of the cycle follows the discrete uniform probability model with each individual probability equal to $(\frac{1}{N})^{k-1}$. Hence, it is to be expected that agent $m$ would continue exploring the cycles of $\mathcal{G}_m^{(t)}$ rationally until she commits the so-called redundant action, after which she stops the exploration. Define a new random variable $Z$ which denotes the number of trials that takes for agent $m$ to exploit the cycles of $\mathcal{G}_m^{(t)}$ until she commits a redundant action. Evidently, $Z$ is a geometrically distributed random variable with the probability of success (traversing a self-loop) equal to $(\frac{1}{N})^{k-1}$. As a result, the expected number of Bernoulli trials that takes for agent $m$ to commit a redundant action is 
\begin{equation}
\mathbb{E}[Z]=\frac{1}{(\frac{1}{N})^{k-1}}=N^{k-1}.
\end{equation}
In turn, this implies that the information gain gets increased $N^{k-1}$-fold by traversing the cycles of length $k-1$ on $\mathcal{G}_m^{(t)}$, e.g., compare to \cref{Eq:Slef_Idempotent}.

It is worth to mention that the kinds of interpretations like the one provided above which are consistent with common sense, potentially amplify the importance of RPCMs as being tools deemed appropriate for modeling agents' underlying mechanisms of judgment, further providing justifications for the use of comparative judgment framework in a broader sense.

Furthermore, the scaled self-similarity property of RPCMs leads to an interesting fractal phenomenon pertaining to RPCMs which occurs on hyperbolic plane, as stated in the following theorem.

\begin{theorem}[Fractal Property]
	\label{Thrm:hyperbolic}
	Under the \blue{PCJ} \ref{PCJ_Axiom} and for the agent $m$, the following hyperbolic self-similarity relation
	\begin{equation}
	\forall k\geq 1 , \forall \mathbf{A} \in \mathcal{A}_{\mathrm{RPCM}}^{(t)}:  \sinh(\mathbf{A}) = \frac{\sinh(N)}{N} \mathbf{A}   
	\end{equation}
	holds for every $t\in \mathcal{T}$. 
\end{theorem}
\begin{proof}
	Using the matrix exponential for $\mathbf{A}$ and by resorting to implications of \Cref{Thrm:power}, we may write
	\begin{equation}
	\label{Eq:Matrix_Exponential}
	\exp(\mathbf{A}) = \sum\limits_{k=0}^{\infty} \frac{\mathbf{A}^k}{k!} = \mathbf{I}_N + \frac{1}{N}(e^N-1) \mathbf{A} .
	\end{equation}
	Furthermore, \cref{Eq:Matrix_Exponential} implies that
	\begin{equation}
	\label{Eq:Matrix_Exponential_Neg}
	\exp(-\mathbf{A}) = \mathbf{I}_N + \frac{1}{N}(e^{-N}-1) \mathbf{A} .
	\end{equation}
	By \cref{Eq:Matrix_Exponential,Eq:Matrix_Exponential_Neg}, we then get
	\begin{equation}
	\exp(\mathbf{A}) - \exp(-\mathbf{A}) = \frac{1}{N}(e^N - e^{-N}) \mathbf{A},
	\end{equation}
	which can be expressed in terms of the sine hyperbolic functions as presented below
	\begin{equation}
	\sinh(\mathbf{A}) = \frac{\sinh(N)}{N} \mathbf{A} .
	\end{equation}
\end{proof}

\begin{corollary}
	\label{Thrm:hyperbolic_Trace}
	Under the \blue{PCJ} \ref{PCJ_Axiom} and for the agent $m$, the following relation
	\begin{equation}
	\forall k\geq 1 , \forall \mathbf{A} \in \mathcal{A}_{\mathrm{RPCM}}^{(t)}:   \mathrm{Tr}[\sinh(\mathbf{A})] =  \sinh(N)  
	\end{equation}
	holds for every $t\in \mathcal{T}$. 
\end{corollary}
\begin{proof}
	It is straightforward.		
\end{proof}

Besides, the above lemmas yield that eigenvalues of an RPCM are either $N$ or $0$ (with algebraic multiplicity $N-1$).
\begin{lemma}[Characteristic Polynomial]
	\label{Thrm:polynomial}
	For the agent $m$ and every $ \mathbf{A} \in  \mathcal{A}_{\mathrm{RPCM}}^{(t)} $ where $t \in \mathcal{T}$, the characteristic polynomial for $\mathbf{A}$ takes on the following form:
	\begin{equation}
	\label{Eq:Characteristic_Polynomial}
	p_{\mathbf{A}}(\lambda ) = \lambda^{N} - N\lambda^{N-1}   .
	\end{equation}
\end{lemma}
\begin{proof}
	It is well-known that the characteristic polynomial of a square matrix is given by
	\begin{equation}
	p_{\mathbf{A}}(\lambda ) = \lambda^{N} - \left[ \mathrm{Tr}  (\mathbf{A}) \right] (\lambda^{N-1}) + \dots + (-1)^N \det(\mathbf{A})  .
	\end{equation}	
	By \Cref{Thrm:Rank}, all the terms of $p_{\mathbf{A}}(\lambda )$ except the first two leading terms equal zero. It remains to note that $ \mathrm{Tr}(\mathbf{A})= N $ for any RPCM, which yields
	\begin{equation}
	p_{\mathbf{A}}(\lambda ) = \lambda^{N} - N (\lambda^{N-1})  .
	\end{equation}
	In turn, this completes the proof.			
\end{proof}

Using the implications of \Cref{Eq:Characteristic_Polynomial}, we provide a necessary condition on the consistency property $(iii)$ of the \blue{PCJ} \ref{PCJ_Axiom} in terms of the \textit{permanents} of the RPCMs, as stated in the following theorem.

\begin{theorem}[Matrix Permanent Property]
	\label{Thrm:Permanent}
	Under the \blue{PCJ} \ref{PCJ_Axiom} and for the agent $m$, the following relation
	\begin{equation}
	\label{Eq:Permanent}
	\forall k\geq 1 , \forall \mathbf{A} \in \mathcal{A}_{\mathrm{RPCM}}^{(t)}:  \mathrm{perm}(\mathbf{A}) = N!  
	\end{equation}
	holds for every $t\in \mathcal{T}$, where $\mathrm{perm}(\mathbf{A})$ represents the permanent of the matrix $\mathbf{A}$.
\end{theorem}
\begin{proof}
	In order to proved the theorem, it is sufficient to demonstrate that the matrices $\mathbf{A}$ and $\mathbf{J}_N$ (all-ones matrix of order $N$) are similar. This assertion simply follows from the fact that both of the matrices mentioned above share a common characteristic polynomial 
	\begin{equation}
	p(\lambda ) = \lambda^{N} - N (\lambda^{N-1})  .
	\end{equation}
	Then, the relation given by \cref{Eq:Permanent} can be shown to hold true due to the fact that $\mathrm{perm}(\mathbf{A}) = N!$ if and only if $\mathbf{A}$ can be obtained from $\mathbf{J}_N$ by a finite sequence of elementary operations, e.g., see \Citep{wang1974permanents,akbari2016permanents} for more technical details.
\end{proof}

Next, we use \cref{Eq:Characteristic_Polynomial} for the characteristic polynomial of $ \mathbf{A}$ to determine the eigenvalues of the matrix.
\begin{lemma}[Perron-Frobenius Eigenvalue]
	\label{Thrm:PF}
	For the agent $m$, it holds that 
	\begin{equation}
	\forall \mathbf{A} \in \mathcal{A}_{\mathrm{RPCM}}^{(t)}, \lambda_{\max} = N   
	\end{equation}
	for every $t\in \mathcal{T}$, where $\lambda_{\max} $ represents the largest eigenvalue of $\mathbf{A}$, also called the Perron-Frobenius eigenvalue.
\end{lemma}
\begin{proof}
	Recall that for a matrix $\mathbf{A}$ with strictly positive elements, the Perron-Frobenius theorem states that $\mathbf{A}$ has a  positive real eigenvalue $\lambda_{\max}$, which is strictly greater than the absolute values (moduli) of all the other eigenvalues.
	On the other hand, by \Cref{Thrm:polynomial} the roots of the characteristic polynomial of $\mathbf{A}$ are either $0$ or $N$. In addition, \Cref{Thrm:Rank} implies that $\mathrm{rank}(\mathbf{A})~=~1$. Hence, $\mathbf{A}$ has a single nonzero eigenvalue $\lambda_{\max}=N$ with all the remaining eigenvalues equal to zeros.	
\end{proof}

An another interesting property of RPCMs is that they belong to a family of matrices for which there is a relation between their rank, order and trace (see \cref{Eq:Tr_Rank}).
\begin{lemma}[Trace-Order-Rank Property]
	\label{Thrm:Idompotent_Class}
	Under the \blue{PCJ} \ref{PCJ_Axiom} and for the agent $m$, the specific characteristics of a generic matrix $\mathbf{A} \in \mathcal{A}_{\mathrm{PCM}}^{(t)}$ satisfies the following relationship: 
	\begin{equation}
	\label{Eq:Tr_Rank}
	\mathrm{Tr}( \mathbf{A} ) = N \boldsymbol{\cdot} \mathrm{rank}( \mathbf{A})  .
	\end{equation}
\end{lemma}
\begin{proof}
	By \Cref{Thrm:power}, we have
	\begin{equation}
	\label{Eq:proposition}
	\mathbf{A}^{k} = (N^{k-1}) \mathbf{A}  , \quad \forall k\geq 1  .
	\end{equation}
	For $k=2$, \cref{Eq:proposition} yields that
	\begin{equation}
	\label{Eq:proposition_2}
	\mathbf{A}^{2} = N \mathbf{A}  .
	\end{equation}
	The proof then follows from the fact that \cref{Eq:proposition_2} implies the validity of \cref{Eq:Tr_Rank}.
	See \Citep{harville2018linear} for more details on this implication.	
\end{proof}

Further on, since for any RPCM we have $\mathrm{Tr}( \mathbf{A} )=\lambda_{\max}$, in such case \Cref{Thrm:Idompotent_Class} stipulates that
\begin{equation}
\label{Eq:Eig_Rank}
\lambda_{\max} = N \boldsymbol{\cdot} \mathrm{rank}( \mathbf{A})  .
\end{equation}

\Cref{Eq:Eig_Rank} can be interpreted as follows: It is known that if we take any vector corresponding to the largest eigenvalue $\lambda_{\max}$, then the action of the matrix on that specific vector attains its maximum value. Next, recall that $\mathrm{rank}( \mathbf{A})$ quantifies the amount of information that is preserved by such action. Hence, \cref{Eq:Eig_Rank} means that the maximum action of the matrix equals its complexity (provided that the order $N$ of the matrix signifies its complexity) multiplied by the amount of information preserved by such action. In turn, this provides an additional useful property for RPCMs and substantiates their use to model the thought processes of particular agents. Namely, \cref{Eq:Eig_Rank} guarantees that the possibility for the loss of information is eliminated inasmuch as action of RPCMs on vectors chosen along the eigenvector corresponding to $\lambda_{\max}$ are taken into account. This is due to the fact since $\mathrm{rank}( \mathbf{A})=1$ for RPCMs, then reducing the dimensionality of the data set would require one to project each data point along some unit vector. Moreover, it is of great interest to choose the unit vector in a way that projection of the data points along it would retain as much of the variation of the data points as possible. To achieve this goal, it is well-known that the eigenvector corresponding to $\lambda_{\max}$ can be chosen for the purpose of dimensionality reduction since such an eigenvector is the direction along which the data set would have the maximum amount of variance.

Next, another important aspect of RPCMs that deserves attention pertains to sensitivity of their eigenvalues to perturbations of their elements. A possible approach for conducting such sensitivity analysis can be carried out by studying the condition number of the matrix since the change in the output of a function for a small changes in its input arguments can be quantified by the condition number of the function. For a simple eigenvalue $\lambda_j$, it was shown in \Citep{wilkinson1965algebraic} that $\mathrm{cond}(\lambda_j;\mathbf{A})$, expressed in terms of the left and right eigenvectors of $\mathbf{A}$, measures the sensitivity of $\lambda_j$ to small perturbations in elements of the matrix. Yet, another useful representation for $\mathrm{cond}(\lambda_j;\mathbf{A})$ is given in \Citep{smith1967condition} whose author puts forward the following relationship
\begin{equation}
\label{Eq:condition_num}
\mathrm{cond}(\lambda_j;\mathbf{A}) = \frac{\lVert \mathrm{adj}( \lambda_j I - \mathbf{A} ) \rVert_2}{p^{'}_{\mathbf{A}}(\lambda_j)} \CommaPunct
\end{equation}
where $\lVert \mathbf{A} \rVert_2$ denotes the \textit{largest singular value} of $\mathbf{A}$, and $p^{'}_{\mathbf{A}}(\lambda_j)$ is the derivative of the characteristic polynomial of $\mathbf{A}$ w.r.t. $\lambda$ which is evaluated at $\lambda=\lambda_j$. For the matrix $\mathbf{A}$, being an RPCM with a simple eigenvalue $\lambda_{\max}=N$, then \cref{Eq:condition_num} simplifies as follows:
\begin{equation}
\label{Eq:condition_num_RPCM}
\mathrm{cond}(\lambda=N;\mathbf{A}) = \frac{\lVert \mathrm{adj}( N I - \mathbf{A} ) \rVert_2}{N^{N-1}}  \cdot
\end{equation}
Therefore, it becomes clear that for a matrix $\mathbf{A} \in \mathcal{A}_{\mathrm{RPCM}}^{(t)}$ with a sufficiently large order $N$, whose majority of the elements are from the same order of magnitude and none of its elements is near-zero or near-infinity, the eigenvalue of the matrix would be relatively insensitive to changes occurred in all of its elements, unless the structure of the matrix undergoes a considerable perturbation. We refer to this property of the largest eigenvalue of RPCMs as \textit{robustness}.
\begin{lemma}[Robustness of Perron-Frobenius Eigenvalue]
	\label{Thrm:Robustness}
	The largest eigenvalue $\lambda_{\max}$ of a generic matrix $\mathbf{A} \in \mathcal{A}_{\mathrm{RPCM}}^{(t)}$ is robust w.r.t. perturbations occurred in the matrix, provided that the elements of the matrix are homogeneous, i.e., they are uniformly bounded from above and below.
\end{lemma}
\begin{proof}
	See \Citep{saaty1993relative}.	
\end{proof}

An application of the sensitivity analysis for RPCMs is illustrated by the following example:
\begin{example}
	\label{Ex:Sensitivity}
	Consider the following matrix $\mathbf{A} \in \mathcal{A}_{\mathrm{RPCM}}^{(t)}$
	\begin{equation*}
	\mathbf{A} = 
	\begin{pmatrix}
	1 & \alpha & \beta \\
	\frac{1}{\alpha} & 1 & \gamma  \\
	\frac{1}{\beta}  & \frac{1}{\gamma}  & 1  \\
	\end{pmatrix}  ,
	\end{equation*}
	where $\alpha,\beta,\gamma>0$. Then, it can be shown that the adjugate of $3 I - \mathbf{A}$ is as follows:
	\begin{equation*}
	\mathrm{adj}( 3 I - \mathbf{A} ) = 
	\begin{pmatrix}
	3 & 2\alpha+\frac{\beta}{\gamma} & 2\beta+\alpha\gamma \\
	\frac{2}{\alpha}+\frac{\gamma}{\beta} & 3 & \frac{\beta}{\alpha}+2\gamma  \\
	\frac{2}{\beta}+\frac{1}{\alpha\gamma}  & \frac{\alpha}{\beta}+ \frac{2}{\gamma}  & 3  \\
	\end{pmatrix}  .
	\end{equation*}
	This enables us to compute $\mathrm{cond}(\lambda=3;\mathbf{A})$. We forego to provide an explicit relation for the condition number due to its excessive length. Now, assume that the original matrix $\mathbf{A}$ is constructed using the triplet $(\alpha,\beta,\gamma)=(1.001,0.995,1.005)$ and the perturbed matrices are created using the triplets $(1.002,0.980,1.015)$, $(0.950,0.775,1.015)$, $(1.875,0.205,0.580)$ and $(5.225,3.170,0.001)$. Then it can be shown that the condition number for the original matrix is $1.000$, and for the perturbed matrices the condition numbers are equal to $1.000$, $1.015$, $2.030$ and $417.709$, respectively. It is then observed that perturbations occurred in the neighborhood of the original matrix do not influence the condition number substantially, whereas an structural change in the matrix, as in the latter case, leads to a huge change in the value of the condition number.
	
	\noindent Besides, the derivation of $\lVert \cdot \rVert_2$ is often too involved and complex. Hence, one can potentially obtain an overestimate of $\mathrm{cond}(\lambda=3;\mathbf{A})$ by resorting to the well-known inequality $\lVert \cdot \rVert_2 \leq \lVert \cdot \rVert_F$. Yet, another conservative upper bound on $\mathrm{cond}(\lambda=3;\mathbf{A})$ is obtained by observing that
	\begin{equation*}
	\mathrm{adj}( 3 I - \mathbf{A} ) = 
	\begin{pmatrix}
	1 & 0 & 0 \\
	0 & 1 & 0  \\
	0  & 0  & 1  \\
	\end{pmatrix} 
	+
	2\begin{pmatrix}
	1 & \alpha & \beta \\
	\frac{1}{\alpha} & 1 & \gamma  \\
	\frac{1}{\beta}  &  \frac{1}{\gamma}  & 1  \\
	\end{pmatrix} 
	+
	\begin{pmatrix}
	1 & \frac{\beta}{\gamma} & \alpha\gamma \\
	\frac{\gamma}{\beta} & 1 & \frac{\beta}{\alpha}  \\
	\frac{1}{\alpha\gamma}  & \frac{\alpha}{\beta}  & 1  \\
	\end{pmatrix}  . 
	\end{equation*}
	
	\noindent Then, $\mathrm{cond}(\lambda=3;\mathbf{A})$ would be bounded from above by summing the Frobenius norms of each decomposed matrix, i.e., an upper bound on $\mathrm{cond}(\lambda=3;\mathbf{A})$ is $$ \left(\frac{1}{N^{N-1}} \right) \max\left\{ \max(N,\alpha,\beta,\gamma,\alpha\gamma,\frac{\beta}{\gamma}) , \frac{1}{\min(\alpha,\beta,\gamma,\alpha\gamma,\frac{\beta}{\gamma})}  \right\} \cdot $$
	In turn, this implies that insofar as the elements of the matrix $\mathbf{A}$ are uniformly bounded, perturbations in the neighborhood of $\mathbf{A}$ do not affect its maximum eigenvalue substantially.
\end{example}

The concluding properties of RPCMs we discuss here pertains to the eigenvalues of their perturbed matrices. Namely, the following lemma holds.
\begin{lemma}
	\label{Thrm:Eig_Grt_N}
	The Perron-Frobenius eigenvalue $\lambda_{\max}=N$ for the matrix $\mathbf{A}$ if and only if the matrix $\mathbf{A} \in \mathcal{A}_{\mathrm{RPCM}}^{(t)}$. Otherwise, $\lambda_{\max}>N$ for any matrix being a perturbation of $\mathbf{A}$. 
\end{lemma}
\begin{proof}
	See \Citep{saaty1993relative}.
\end{proof}

Furthermore, the following lemma delineates the relation between Perron-Frobenius eigenvalues of RPCMs and their consistency properties (property ($iii$) of the \blue{PCJ} \ref{PCJ_Axiom}).
\begin{lemma}
	\label{Thrm:Consistency}
	The matrix $\mathbf{A}\in \mathcal{A}_{\mathrm{RPCM}}^{(t)}$ for some $t\in \mathcal{T}$, is consistent if and only if the Perron-Frobenius eigenvalue of the matrix is $\lambda_{\max}=N$.
\end{lemma}
\begin{proof}
	See \Citep{saaty1993relative}.
\end{proof}

As per \Cref{Thrm:Eig_Grt_N,Thrm:Consistency}, any perturbation of the matrix $\mathbf{A} \in \mathcal{A}_{\mathrm{RPCM}}^{(t)}$ can be seen as a departure from the fulfillment of the condition $(iii)$ of the \blue{PCJ} \ref{PCJ_Axiom}. This is due to the fact that $\lambda_{\max}=N$ for any $\mathbf{A}\in \mathcal{A}_{\mathrm{RPCM}}^{(t)}$ at every $t\in \mathcal{T}$. Therefore, at any given point of time $t$, one may write the definition of consistency as follows:
\begin{equation}
a_{il}^{(t)} a_{lj}^{(t)} = a_{ij}^{(t)}, \quad i,j,l=1,2,\dots,N.
\end{equation}
This implies that the average degree of consistency on $\mathcal{G}_m^{(t)}$ for the agent $m$ yields $\mathrm{cDeg}(\mathcal{E}_m^{(t)})=1$ using \cref{Eq:cDeg} for all $t\in \mathcal{T}$. In turn, this permits the agent $m$ to quantify her judgments to an extent that she can perform associative comparisons among all $N^3$ triplets of assets.

%%%%%%%%%%%%%%%%%%%%%%%%%%%%%%%%%%%%%%%%%%%%%%%%%%%%%%%%%%%%%%%%%%%%%%%%%%%%%%%%%%%%%%%%%%%%%%%%%%%%%%%%%
%%%%%%%%%%%%%%%%%%%%%%%%%%%%%%%%%%%%%%%%%%%%%%%%%%%%%%%%%%%%%%%%%%%%%%%%%%%%%%%%%%%%%%%%%%%%%%%%%%%%%%%%%

\subsection{Pairwise Comparison Tensors}
\label{Section:2_2}
In this section, we employ tensors as a natural structure to embody the information pertaining to the mutual interactions among price changes in the assets. Let us start by defining the so-called \textit{pairwise comparison tensors} and their \textit{reciprocal} counterparts.
\begin{definition}[PCT]
	\label{Def:Comparison_Tensor}
	A third-order spatio-temporal tensor $\mathbfcal{A} \in  \mathcal{A}_{\mathrm{PCT}}^{|\mathcal{T}|} \subset \mathbb{R}_{> 0}^{N\times N\times |\mathcal{T}|}$ is said to be a \textit{pairwise comparison tensor} (PCT) if it is formed by concatenating $|\mathcal{T}|$ number of PCMs as its frontal slices, where $\mathcal{A}_{\mathrm{PCT}}^{|\mathcal{T}|}$ denotes the family of PCTs having temporal dimensions $|\mathcal{T}|$.
\end{definition}

\begin{definition}[RPCT]
	\label{Def:R_Comparison_Tensor}
	A generic tensor $\mathbfcal{A} \in  \mathcal{A}_{\mathrm{RPCT}}^{|\mathcal{T}|}$ is said to be a \textit{reciprocal pairwise comparison tensor} (RPCT) if each of its frontal slices is an RPCM, where $\mathcal{A}_{\mathrm{RPCT}}^{|\mathcal{T}|}$ denotes the family of RPCTs having temporal dimensions $|\mathcal{T}|$, which is a proper subclass of $\mathcal{A}_{\mathrm{PCT}}^{|\mathcal{T}|}$.
\end{definition}

In what follows, we extend the concept of matrix inconsistency introduced in \Citep{saaty1977scaling,saaty1993relative} to the domain of tensors and define a function to measure the inconsistencies of PCTs. Keeping this in mind, recall \Cref{Thrm:Eig_Grt_N} which states that the largest eigenvalue $\lambda_{\max}$ of any matrix obtained by perturbing its associated $\mathbf{A} \in \mathcal{A}_{\mathrm{RPCM}}^{(t)}$ would have its largest eigenvalue being greater than $N$. More so, it was discussed that the consistency property $(iii)$ of the \blue{PCJ} \ref{PCJ_Axiom} is satisfied only if a given matrix belongs to $\mathcal{A}_{\mathrm{RPCM}}^{(t)}$ (e.g., see \Cref{Thrm:Consistency}). Hence, one may envisage the statistic $\lambda_{\max}-N$ as being a measure of departure of the judgments cast by the agent $m$ from the consistency condition of the \blue{PCJ} \ref{PCJ_Axiom}. An average perturbation statistic can then be defined as $(\lambda_{\max}-N)/(N-1)$.
\begin{definition}[Inconsistency Function]
	\label{Def:Incons_Func}
	For a tensor $\mathbfcal{A} \in  \mathcal{A}_{\mathrm{PCT}}^{|\mathcal{T}|}$, the \textit{inconsistency function} $\psi: \mathcal{T} \to \mathbb{R}_{\geq 0}^{|\mathcal{T}|}$ is defined as follows:
	\begin{equation}
	\label{Eq:inconsistency}
	\psi(t) := \frac{\lambda^{(t)}_{\max} - N}{N - 1}\cdot		
	\end{equation}
	Here, $\lambda^{(t)}_{\max}$ represents the largest eigenvalue of the PCM corresponding to the $t$-th frontal slice of~$\mathbfcal{A}$.
\end{definition}

In turn, \Cref{Def:Incons_Func} suggests one to examine temporal inconsistencies which agent $m$ may confront in the protean of making her judgments. It can be stated that the judgments cast by the agent $m$ are temporally consistent if her judgments at every point of time $t \in \mathcal{T}$ satisfy the condition ($iii$) of the \blue{PCJ} \ref{PCJ_Axiom}. That is, the inconsistency function for a pairwise comparison tensor $\mathbfcal{A} \in  \mathcal{A}_{\mathrm{PCT}}^{|\mathcal{T}|}$ is zero if every frontal slice of $\mathbfcal{A}$ is an RPCM. In other words, in an ideal case for which the reciprocal pairwise comparison tensors are used for agent $m$'s absorption of the market information, one can ensure that the mechanisms underlying the agent's subjective opinion remain persistent over time, i.e., $\mu_m^{(t)} = \mu_m $ for all $t \in \mathcal{T}$ (one may find this situation similar to the scenario involving stationary frequencies of an ergodic Markov chain).

Such regularity of the inconsistency measures obtained for a PCT can also be defined by using the concept of \textit{approximate entropy} (ApEn). In order to define the ApEn for a sequence of inconsistency measures $\{\psi(1),\psi(2),\dots\}$  of length $|\mathcal{T}|$, firstly two blocks of longitude $l\leq |\mathcal{T}|$ are defined by $\breve{\psi}(i)=\{\psi(i),\dots,\psi(i+l-1)\}$ and $\breve{\psi}(j)=\{\psi(j),\dots,\psi(j+l-1)\}$, and then the distance between them is calculated using $d(\breve{\psi}(i),\breve{\psi}(j))=\max\limits_{k=1,\dots,l}(|\psi(i+k-1)-\psi(j+k-1)|)$. Thereafter, the following statistic
\begin{equation}
C^l_i(r) = \frac{1}{|\mathcal{T}|-l+1} \sum\limits_{j=1}^{|\mathcal{T}|-l+1} \mathbb{I}[d(\breve{\psi}(i),\breve{\psi}(j))\leq r ]
\end{equation}
is defined, where the summation term counts the number of consecutive blocks of longitude $l$ which are similar to the given block $\breve{\psi}(i)$ within a given resolution~$r$. Subsequently, the logarithmic frequency with which longitude-$l$ blocks that are close together stay together for the next increment defines the ApEn for the given sequence of inconsistency measures, i.e.,
\begin{equation}
\label{Eq:App_Ent}
\mathring{\psi} = \frac{1}{|\mathcal{T}|-l+1} \sum\limits_{i=1}^{|\mathcal{T}|-l+1} C^l_i(r).
\end{equation}
ApEn defined by \cref{Eq:App_Ent} enjoys several mathematical and statistical properties among which the following are noteworthy: 1- ApEn is a statistic which is robust, meaning that it is insensitive to artifacts or outliers; 2- ApEn is not altered by translations or scaling applied uniformly to all elements contained within the considered sequence of inconsistency measures; 3- nonlinearity causes greater values of ApEn, e.g, see \Citep{pincus1994physiological}; 4- ApEn is model-free, which makes it a promising statistic for the analysis of data series whose generating sources are unknown; and 5- Computation of ApEn requires equally-spaced measurements over time, e.g, see \Citep{pincus1991approximate}.

The above-mentioned properties of ApEn, among others, make it a suitable candidate to measure the regularity of the inconsistency measures. This has been formally defined as follows:

\begin{definition}[Regularity Measure]
	\label{Def:Incons_Measure}
	For a tensor $\mathbfcal{A} \in  \mathcal{A}_{\mathrm{PCT}}^{|\mathcal{T}|}$, the regularity measure of inconsistencies $\mathring{\psi}: \mathbb{R}_{\geq 0}^{|\mathcal{T}|} \to \mathbb{R}_{\geq 0}$  is defined by the approximate entropy of the inconsistency function $\psi(t)$ given by \cref{Eq:inconsistency}.
\end{definition}

By appealing to the applications of $\mathring{\psi}$, one can then quantify the regularity of the inconsistencies, e.g., tackling the question of whether a pairwise comparison tensor $\mathbfcal{A} \in  \mathcal{A}_{\mathrm{PCT}}^{|\mathcal{T}|}$ is regularly inconsistent (consistent), or it is irregularly inconsistent. Basically, the regularity measure $\mathring{\psi}$ of inconsistencies estimates the amount of randomness found in $\psi(t)$ without requiring any prior knowledge of the source that generates the time-dependent inconsistency values. In other words, lower values of $\mathring{\psi}$ attest that $\psi(t)$ is persistent, repetitive and predictive in the sense that patterns repeat themselves throughout the series. On the other hand, higher values of $\mathring{\psi}$ imply independence between particular temporal values of $\psi(t)$ and hence are indicators towards lower number of repeated patterns and smaller values of randomness, e.g., see \Citep{delgado2019approximate,pincus2008approximate} for more details on approximate entropy. 
\begin{remark}
	A pairwise comparison tensor $\mathbfcal{A} \in  \mathcal{A}_{\mathrm{PCT}}^{|\mathcal{T}|}$ is said to be regularly consistent if
	\begin{equation}
	(\mathring{\psi} = 0) \wedge (\forall t: \psi(t) = 0) ,
	\end{equation}
	and regularly inconsistent if
	\begin{equation}
	(\mathring{\psi} = 0) \wedge (\exists t: \psi(t) > 0)  .
	\end{equation}
	Otherwise, $\mathbfcal{A}$ is said to be irregularly inconsistent if
	\begin{equation}
	(\mathring{\psi} > 0) \wedge (\exists t: \psi(t) > 0)  ,
	\end{equation}
	where the magnitude of $\mathring{\psi}$ delineates the amount of irregularity.
\end{remark}

%%%%%%%%%%%%%%%%%%%%%%%%%%%%%%%%%%%%%%%%%%%%%%%%%%%%%%%%%%%%%%%%%%%%%%%%%%%%%%%%%%%%%%%%%

The RPCTs constructed for agent $m$ according to  \Cref{Def:Comparison_Tensor}, reflect the ideal situations in which the agent satisfies the conditions of the \blue{PCJ} \ref{PCJ_Axiom} throughout time. In reality, however, the values of the comparison tensor could deviate substantially from the ideal case. In order to illustrate this phenomenon, consider tensor $\mathbfcal{A} \in  \mathcal{A}_{\mathrm{RPCT}}^{|\mathcal{T}|}$ whose perturbed tensor is denoted by $\tilde{\mathbfcal{A}}\in  \mathcal{A}_{\mathrm{PCT}}$ and referred to as the \textit{consensus pairwise comparison tensor}. The following theorem shows that $ \tilde{\mathbfcal{A}}$ has to be a rank-$1$ estimate of $\mathbfcal{A}$ to guarantee that the inconsistency function \eqref{Eq:inconsistency} can be recovered.
\begin{theorem}[Consensus PCT]
	\label{Thrm:consensus}
	Let $\mathbfcal{A} \in  \mathcal{A}_{\mathrm{RPCT}}^{|\mathcal{T}|}$ be a tensor whose associated consensus tensor is denoted by $\tilde{\mathbfcal{A}}\in  \mathcal{A}_{\mathrm{PCT}}$. Also, let $\tilde{\mathbfcal{A}}$ be a rank-$1$ tensor, i.e., $\tilde{\mathbfcal{A}}=~\tilde{\mathbf{z}}~\circ~( \tilde{\mathbf{x}} \circ \tilde{\mathbf{y}}^\intercal )$, which solves the following constrained polyadic decomposition model:
	\begin{subequations}
		\label{Eq:Decomposition}
		\begin{align}
		\min_{\substack{\tilde{\mathbf{x}},\tilde{\mathbf{y}},\tilde{\mathbf{z}}}} &\,\,   \left\Vert \mathbfcal{A} - \tilde{\mathbf{z}}  \circ  ( \tilde{\mathbf{x}} \circ \tilde{\mathbf{y}}^\intercal ) \right\Vert_F   \\
		\mathrm{s.t.} &\,\,   \tilde{z}_t   ( \tilde{\mathbf{x}} \circ \tilde{\mathbf{y}}^\intercal ) >\mathbf{0}_N  ,  \quad t=1,2,\dots ,|\mathcal{T}|   , \label{Eq:Decomposition_1} \\ 
		&\,\,   \tilde{\mathbf{x}}, \tilde{\mathbf{y}} \in \mathbb{R}^N_{>0}   ,  \\
		&\,\,   \tilde{\mathbf{z}} \in \mathbb{R}^{|\mathcal{T}|}_{\geq 0}   .  
		\end{align}
	\end{subequations}
	Here, $\mathbf{0}_N$ denotes an $N\times N$ matrix whose entries are all zeros, and inequality \eqref{Eq:Decomposition_1} is evaluated component-wise. Then, the elements of the vector $\tilde{\mathbf{z}}$ depend on the largest eigenvalues of the frontal slices of $\tilde{\mathbfcal{A}}$ (up to some positive scaling).
\end{theorem}
\begin{proof}
	Consider the matrix $\mathbf{A} \in \mathcal{A}_{\mathrm{RPCM}}^{(t)}$ at a given time $t\in\mathcal{T}$, and let $\mathbf{u}$ and $\mathbf{v}$ represent the left and right eigenvectors of its associated Perron-Frobenius eigenvalue $\lambda_{\mathrm{max}}$, respectively, i.e., 
	\begin{equation*}
	\mathbf{u}^\intercal \mathbf{A} = \lambda_{\mathrm{max}} \mathbf{u}^\intercal \quad \text{and} \quad \mathbf{A} \mathbf{v}  = \lambda_{\mathrm{max}} \mathbf{v}  ,
	\end{equation*}
	where strict positivity of $\mathbf{A}$ renders $\mathbf{u}, \mathbf{v} \in \mathbb{R}^N_{>0}$. Then, by Perron-Frobenius theorem we have 
	\begin{equation}
	\label{EQ:Perron-Frobenius}
	\lim\limits_{k \to \infty} \left(\dfrac{\mathbf{A}}{\lambda_{\mathrm{max}}}\right)^k = \left( \frac{1}{\mathbf{v}^\intercal \mathbf{u}} \right) \mathbf{u} \mathbf{v}^\intercal  \cdot
	\end{equation}
	However, \Cref{Thrm:power} implies that 
	$
	\mathbf{A}^{k} = N^{k-1} \mathbf{A}  .
	$
	Hence,
	\begin{equation}
	\mathbf{A} \sim \left( \frac{\lambda_{\mathrm{max}}^k}{N^{k-1}} \right) \left( \frac{1}{\mathbf{v}^\intercal \mathbf{u}} \right) \mathbf{u} \mathbf{v}^\intercal = \lambda_{\mathrm{max}} \left( \frac{1}{\mathbf{v}^\intercal \mathbf{u}} \right) \mathbf{u} \mathbf{v}^\intercal  ,
	\end{equation}
	which implies that $\mathbf{A}$ is a rank-$1$ matrix multiplied by a constant which depends on its largest eigenvalue $\lambda_{\mathrm{max}}$.
	
	\noindent On the other hand, the rank-$1$ polyadic decomposition of the tensor $\mathbfcal{A}\in \mathcal{A}_{\mathrm{RPCT}}^{|\mathcal{T}|}$  can be written as the outer product of a rank-$1$ matrix by a vector which resides on the temporal dimension of the tensor, i.e.,
	\begin{equation}
	\mathbfcal{A} \approx \tilde{\mathbfcal{A}} = \tilde{\mathbf{z}}  \circ  ( \tilde{\mathbf{x}} \circ \tilde{\mathbf{y}}^\intercal )  ,
	\end{equation}
	where $\tilde{\mathbf{x}}, \tilde{\mathbf{y}} \in \mathbb{R}^N_{>0}$ and $\tilde{\mathbf{z}} \in \mathbb{R}^{|\mathcal{T}|}_{\geq 0}$. Therefore, elements of the vector $\tilde{\mathbf{z}}$ can be regarded as the terms which depend on the largest eigenvalues of the corresponding matrices which pertain to frontal slices of the consensus tensor $\tilde{\mathbfcal{A}}$. 
	
	\noindent Consequently, in order to recover the inconsistency function $\psi(t)$ at each time instant, one can first obtain $\tilde{\mathbfcal{A}}$ by solving model \eqref{Eq:Decomposition} and then proceed by computing the largest eigenvalue of the approximated tensor at each of its frontal slices, and thereafter evaluate the inconsistency function given by \cref{Eq:inconsistency}.		
\end{proof}

The following remark elaborates on the relation between the inconsistency function $\psi(t)$ and polyadic decomposition of RPCTs.
\begin{remark}
	In order to recover the inconsistency measures, one may first approximate the tensor $\mathbfcal{A}\in \mathcal{A}_{\mathrm{RPCT}}^{|\mathcal{T}|}$ by solving model~\eqref{Eq:Decomposition} and then proceed by computing the largest eigenvalue of the approximated tensor $\tilde{\mathbfcal{A}}$ at each of its frontal slices. Subsequently, the inconsistency measure at each time instant is derived by evaluating the function $\psi(t)$ given by \eqref{Eq:inconsistency}.
\end{remark}

It is further noted that the error of approximation in rank-$1$ polyadic decomposition of an RPCT is directly linked to its average inconsistency measure, as elucidated in the following example.
\begin{example}
	\label{Example}
	Let $\mathbf{r}^{(t)}\in R_{>0}^{N}$ denote the vector of lag-$l$ daily rates of return for $N$ assets, i.e., the lag-$l$ rates of return for assets are computed using the following relation:
	\begin{equation}
	r^{(t)} = \frac{C^{(t)}}{C^{(t-l)}}\CommaPunct
	\end{equation}
	where $C^{(t)}$ denotes the adjusted closing price for the assets at time $t\in \mathcal{T}$.
	
	\noindent In our simulations, we considered $N=811$ assets, selected to be those that have ever appeared in the list of $S\&P500$ at a time period from January $1990$ to December $2019$ for which their price information did not contain any missing values\footnote{Data were retrieved from the Center for Research in Security Prices (CRSP) provided by Wharton Research at the University of Pennsylvania (WRDS).}. Also, in our datasets, the lag-$1$ rates of return were comprised of $|\mathcal{T}|=7558$ and $|\mathcal{T}|=360$ daily and monthly time steps, respectively.
	
	\noindent The RPCTs were constructed for each given lag value (shown schematically in \Cref{Fig:Comparison_Tensor}), which was then followed by computation of their associated inconsistency function $\psi(t)$, average inconsistency measure $\bar{\psi}(t)$ as well as their regularity measure $\mathring{\psi}(t)$\footnote{All models were solved on Amazon Web Services (AWS) using a server with Xen hypervisor, $128$ cores and $4$TB of RAM memory.}. The values for the lag parameter $l$ we considered in our simulations ranged from $1$ through $10$. 
	
	\begin{figure}
		\centering
		\includegraphics[scale=0.2]{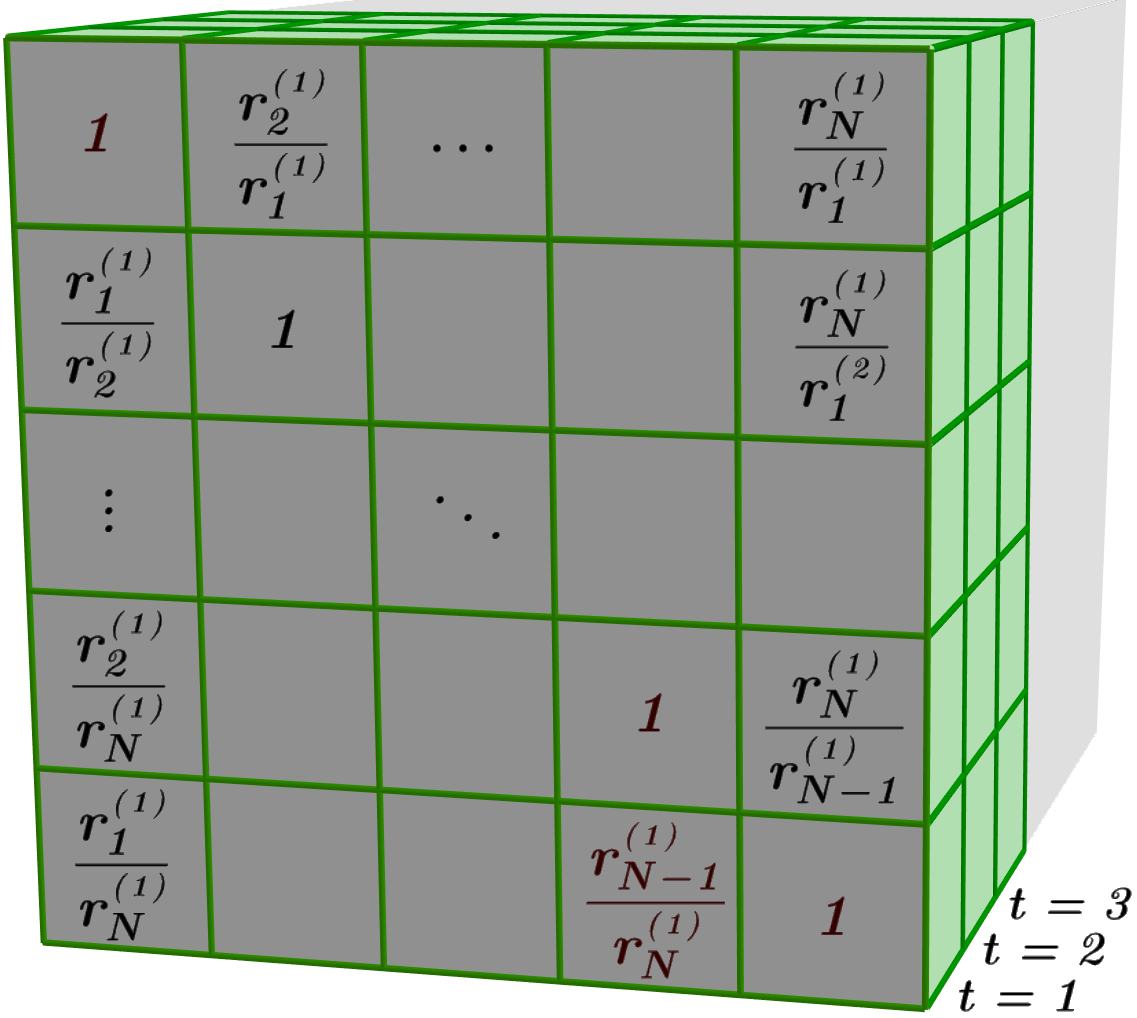}
		\caption[Schematic of a reciprocal pairwise comparison tensor.]{Schematic of an RPCT.}
		\label{Fig:Comparison_Tensor}
	\end{figure}
	
	\noindent \Cref{Fig_Error,Fig_Inconsistency,Fig_Regularity} depict, respectively, the rank-$1$ poliadic decomposition errors, average inconsistency measures and regularity measures for different values of the lag parameter $l$. It is clear from these figures that the error of polyadic decomposition grows nonlinearly by increasing values of the lag parameter, whereas the growth of the average inconsistency measures exhibits a linear trend. These observations were also verified empirically, and led to the formulation of the following relations:
	\begin{equation}
	\epsilon(l) \approx \epsilon(1) \,  \sqrt{l}  ,
	\end{equation}
	and
	\begin{equation}
	\bar{\psi}(l) \approx \bar{\psi}(1)  \, l  , 
	\end{equation}
	where $\epsilon(l)$ and $\bar{\psi}(l)$ denote the relative error of approximation and average inconsistency measure at lag $l$, respectively. Furthermore, \Cref{Fig_Regularity} implies that $\mathring{\psi}(t)$ decreases monotonically as the value of the lag parameter increases.
	
	% In order to shed some light on the underlying reasons for such behaviors, let us first observe that 
	%\begin{equation}
	%	\mathbfcal{A}(l) = \mathbfcal{A}(l) \ast \mathbfcal{A}(l-1) \ast \cdots  \ast \mathbfcal{A}(2) \ast \mathbfcal{A}(1)  , 
	%\end{equation}	
	%for all $l=1,2,\cdots,10$, where $\mathbfcal{A}(l)$ denotes an RPCT constructed using lag-$l$ daily rates of return. Besides, it is known that for two real-valued matrices $\mathbf{B}$ and $\mathbf{C}$ of order $N\times N$, the following inequalities hold true:
	%\begin{equation}
	%	\lVert \mathbf{B} \ast \mathbf{C} \rVert_F \leq \mathrm{Tr}(\mathbf{B}\mathbf{C}^\intercal) \leq \lVert \mathbf{B} \rVert_F \lVert \mathbf{C} \rVert_F  .
	%\end{equation}

	\begin{figure}
		\centering
		\includegraphics[width=0.9\textwidth,height=0.3\textheight]{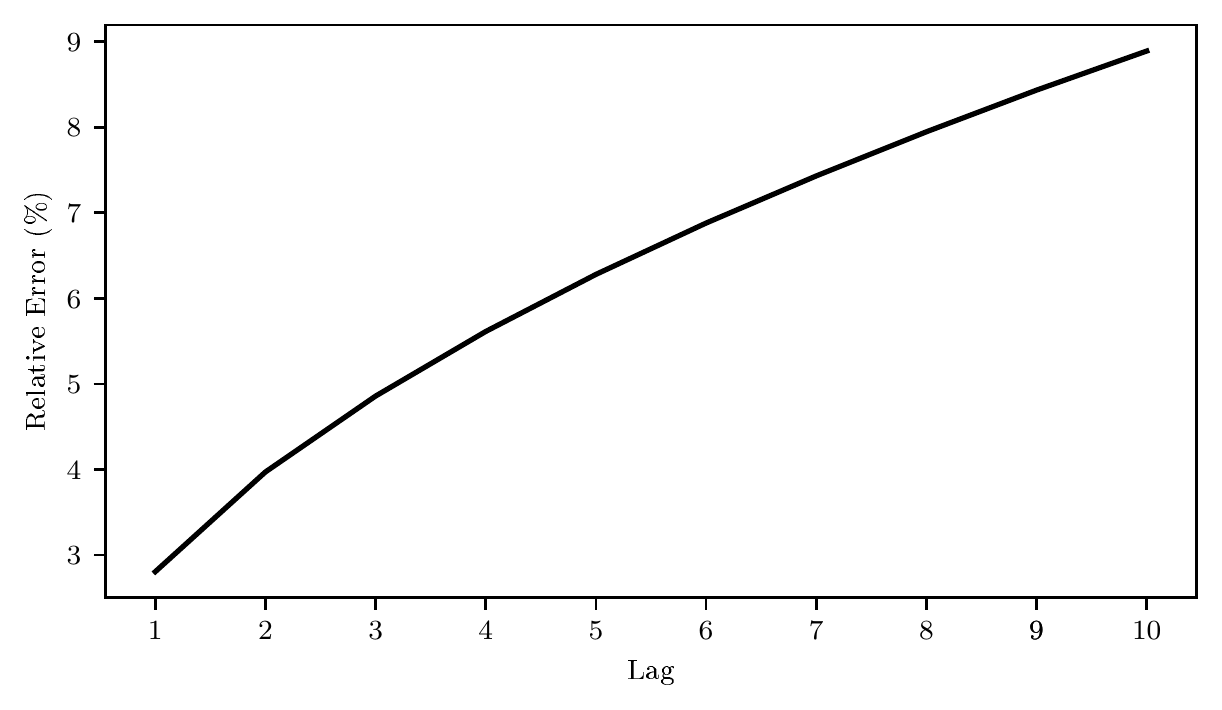}
		\caption[Plot of the relative errors for rank-$1$ polyadic decomposition.]{Plot of the relative errors for rank-$1$ polyadic decomposition at different lag values.}
		\label{Fig_Error}
		\vspace{0.1cm}
		\includegraphics[width=0.9\textwidth,height=0.3\textheight]{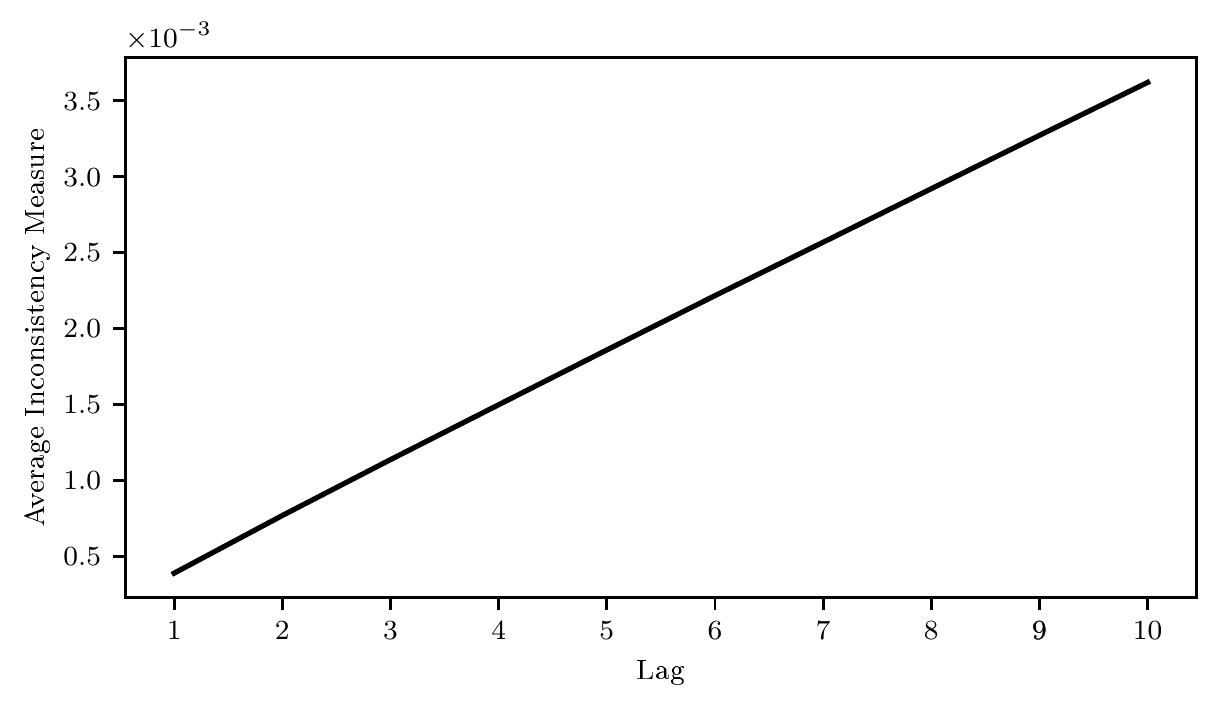}
		\caption[Plot of the average inconsistency measures.]{Plot of the average inconsistency measures at different lag values.}
		\label{Fig_Inconsistency}
		\vspace{0.1cm}
		\includegraphics[width=0.9\textwidth,height=0.3\textheight]{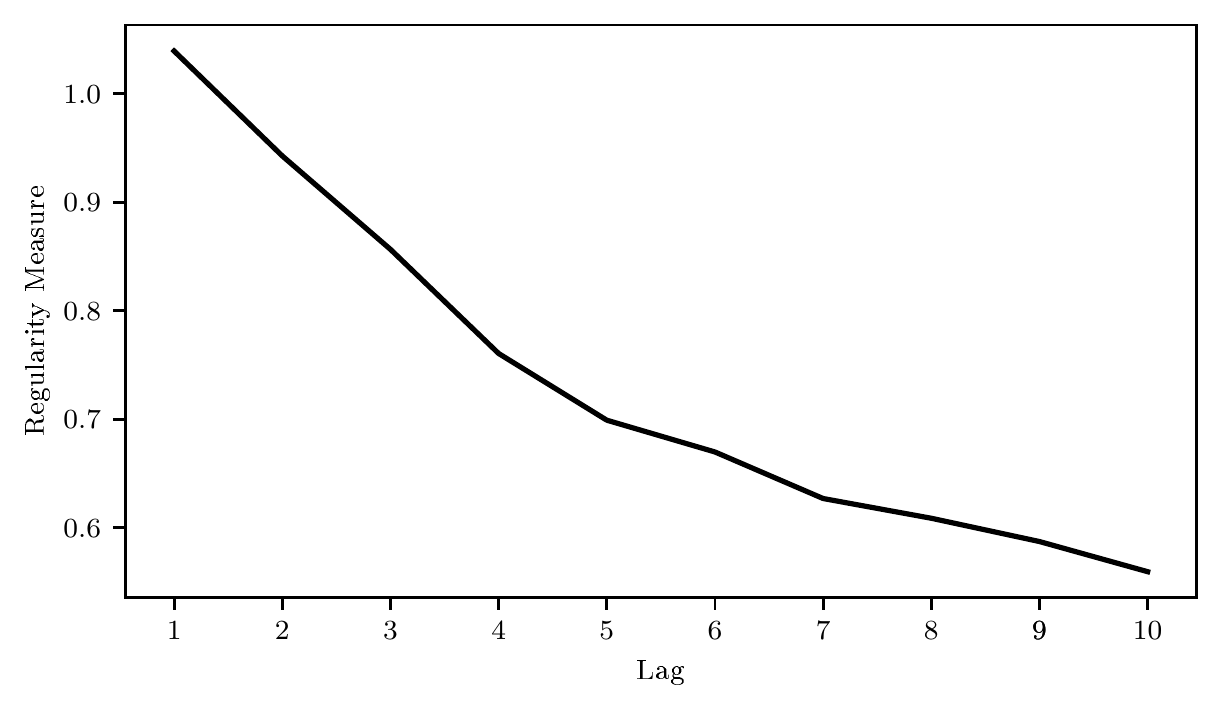}
		\caption[Plot of the regularity measures.]{Plot of the regularity measures at different lag values.}
		\label{Fig_Regularity}
	\end{figure}
	
\end{example}

In the above discussion, we analyzed the mechanisms of judgment casting for a specific agent in the market. However, our ultimate goal is to expand our analysis to the entire participants of the market. For this purpose, we are to make a reasonably realistic assumption based on which the closing prices of assets are the only public source of information available or being of interest to the agents. Then, our analytical framework developed for an individual agent can be extended to encompass the totality of the agents by noting that, for an individual agent, what distinguishes different RPCTs from each other pertains to the lag criterion used in the computation of the rates of return. Thus, the problem of multiple subjective criteria existing across the market (as there are $M\gg 1$ agents who construct their RPCTs using their particular lag values), can potentially be reduced to that of finding an appropriate value for $l$ which represents the lag value used by majority of the agents. Let us call the criterion constructed above, which is based on such a lag value the \textit{constitution} criterion (denoted by $\mathcal{C}$); we adapt the term constitution from \Citep{gibbard2014intransitive} where the author addresses Arrow dilemma \Citep{arrow2012social} and refers to amalgamation of individual preferences as a constitution.

We reckon that the best constitution criterion is the one corresponding to the lag value using which the constructed consensus PCT yields smallest value of the average inconsistency $\bar{\psi}(t)$ and the largest value of the regularity of inconsistencies $\mathring{\psi}(t)$. This is because larger values of $\mathring{\psi}(t)$ indicate higher independence among inconsistencies of PCMs at different points of the time horizon. Thus, opting for the largest value of $\mathring{\psi}(t)$ ensures that $\mathcal{G}_{\mathcal{C}}^{(t)}$ under consideration will more likely to remain independent at different time instants, making it possible to draw more accurate conclusions on relations existing among the assets at intersequent time instants.

%%%%%%%%%%%%%%%%%%%%%%%%%%%%%%%%%%%%%%%%%%%%%%%%%%%%%%%%%%%%%%%%%%%%%%%%%%%%%%%%%%%%%%%%%%%%%%%%%%%%%%%%%
%%%%%%%%%%%%%%%%%%%%%%%%%%%%%%%%%%%%%%%%%%%%%%%%%%%%%%%%%%%%%%%%%%%%%%%%%%%%%%%%%%%%%%%%%%%%%%%%%%%%%%%%%

\subsection{An Index for Quantifying the Chaos in Financial Markets}
\label{Subsection:2_3}
By employing the concepts and tools presented in \Cref{Section:1_1,Section:1_2}, we introduce a stock market index which captures chaotic behavior existing in asset prices. Let us recall that the consensus PCT derived by using some appropriately chosen constitution criterion could potentially embody the stock market information about mutual asset price changes in an effective way. In turn, this enables us to formally define the financial chaos index as follows:
\begin{definition}[FCIX]
	Consider a consensus PCT given by $ \tilde{\mathbfcal{A}}\in~ \mathcal{A}_{\mathrm{PCT}}^{|\mathcal{T}|}$ whose associated RPCT (denoted by $ \mathbfcal{A}\in \mathcal{A}_{\mathrm{RPCT}}^{|\mathcal{T}|}$) is constructed using lag-$l$ rates of return for some value of $l$ as its constitution criterion. The financial chaos index $\mathrm{FCIX}: \mathcal{T} \to \mathbb{R}_{\geq 0}^{|\mathcal{T}|}$ is then defined via the inconsistency function given by \cref{Eq:inconsistency} as follows:
	\begin{equation}
	\mathrm{FCIX}(t) := \psi(t) = \frac{\lambda^{(t)}_{\max} - N}{N - 1}\CommaPunct		
	\end{equation}
	where $\lambda^{(t)}_{\max}$ represents the largest eigenvalue of the matrix corresponding to the $t$-th frontal slice of $\tilde{\mathbfcal{A}}$. We further denote by $\mathrm{FCIX}_t$ (or by $\psi_t$) the time series which takes on the observed FCIX values. 
\end{definition}

For instance, consider the data stated in \Cref{Example} and recall that the lag-$1$ rate of return was deemed suitable for the purpose of constitution criterion. Using the lag-$1$ rates of return for computing $\mathrm{FCIX}_t$, we sketch the plot of annotated monthly $\mathrm{FCIX}_t$ in \Cref{Fig_MFCIX} for a time period from January 1990 to December 2019, respectively. Note that we obtain the monthly  $\mathrm{FCIX}_t$ by taking average of the daily  $\mathrm{FCIX}_t$ values during each month. As it is seen from these figures, the $\mathrm{FCIX}_t$ spikes whenever the stock market undergoes chaotic fluctuations in response to various anomalous political, geopolitical, economical, financial, fiscal, health and psychological events. For instance, such spikes were observed during the first and second Gulf wars, the explosion of the dot-com bubble, September 11, the $2002$ stock market crash, SARS coronovirus pandemic in $2003$, failure of Lehman Brothers, the debt-ceiling dispute in $2011$, Chinese stock market crash in $2016$, OPEC oil cut, the $2018$ world-wide stock market down fall, and the recent US-China trade tensions, to mention but a few. Also, for practical reasons addressed in the sequel, we utilize the notion of the quarterly  $\mathrm{FCIX}_t$, which is defined as the average of daily $\mathrm{FCIX}$ values during each quarter.	
\begin{figure}
	\centering
	%	\resizebox{1\textwidth}{0.4\textheight}{\input{./Figures/Chaos_Index/Fig_FCI.pgf}}
	%	\caption[Plot of the daily $\mathrm{FCIX}_t$.]{Plot of the daily $\mathrm{FCIX}_t$ for a time period from January 1990 to December 2019.}
	%	\label{Fig_FCIX}
	%	\vspace{0.5cm}
	\includegraphics[width=1\textwidth,height=0.4\textheight]{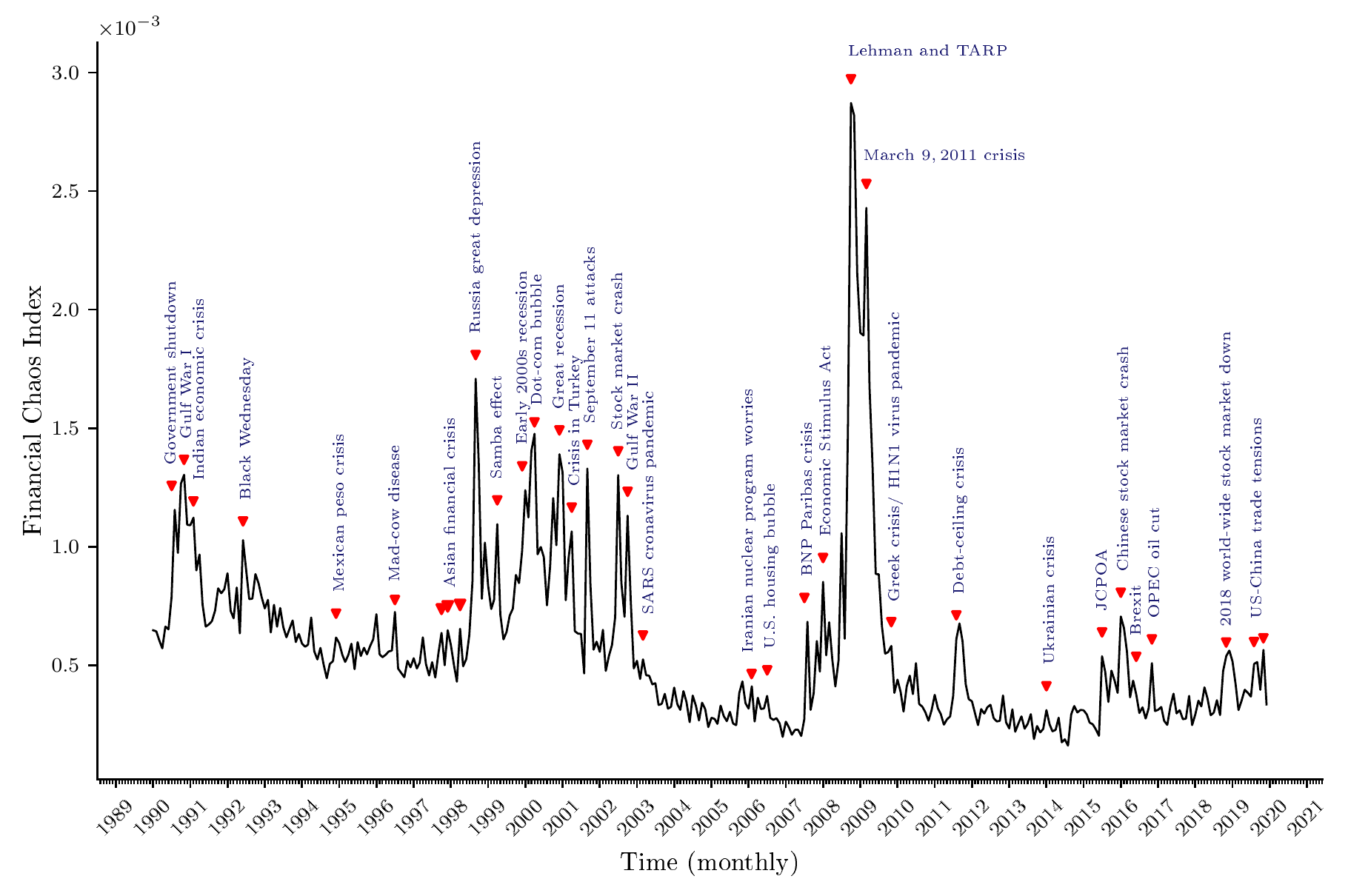}		
%	\resizebox{1\textwidth}{0.4\textheight}{\input{./Figures/Chaos_Index/Fig_MFCI_Annotated.pgf}}
	\caption[Plot of the annotated monthly $\mathrm{FCIX}_t$.]{Plot of the annotated monthly $\mathrm{FCIX}_t$ portrayed during January 1990-December 2019.}
	\label{Fig_MFCIX}	
\end{figure}

The term volatility is a notoriously slippery concept and bears an intricate nature and meaning which makes its quantification a daunting task. One widely-conventional definition of the volatility ascribes it to the amount of \textit{dispersion} of returns for a given set of assets. However, we argue here that the dispersion statistic and its variants are not the only available choices for quantification of volatility, and our developed financial chaos index is in fact an alternative measure for the market volatility, which is superior to some dispersion-based measures from a certain point of view.

Let us consider the tensor $ \mathbfcal{A}\in~ \mathcal{A}_{\mathrm{RPCT}}^{|\mathcal{T}|}$ constructed using some constitution criterion $\mathcal{C}$ as a model for embedding the stock market information whose associated comparison multigraph at each time instant $t\in \mathcal{T}$ is given by $\mathcal{G}^{(t)}_\mathcal{C}=(\mathit{S},\mathcal{E}^{(t)}_\mathcal{C})$. Further, recall that the value of $\mathrm{FCIX}_t$ would become zero for some $t=t_0$ only if the $t_0$-th frontal slice of $\mathbfcal{A}$ would satisfy the \blue{PCJ} \ref{PCJ_Axiom}. In other words, the fulfillment of the consistency property of the \blue{PCJ} \ref{PCJ_Axiom} would be equivalent to existence of a set of edge weights $\mathcal{E}^{(t_0)}_\mathcal{C}$ based on which all activities on $\mathcal{G}^{(t_0)}_\mathcal{C}$ are interconnected. 

Besides, it was shown in \Cref{Section:2_2} that small perturbations in the neighborhood of $\mathbfcal{A}$ would have negligible impact on the value of $\mathrm{FCIX}_{t_0+dt}$, implying presence of a similar set of interconnecting patterns on $\mathcal{G}^{(t_0+dt)}_\mathcal{C}$, a feature also known as near-consistency. On the contrary, a relatively large value of $\mathrm{FCIX}_{t_0+dt}$ as compared to $\mathrm{FCIX}_{t_0}$ would indicate a major structural change in $\mathbfcal{A}$ for the time period $[t_0,t_0+dt]$. Hence, drastically different interconnecting patterns would rather appear at time $t_0+dt$ as compared to those delineated by $\mathcal{E}^{(t_0)}_\mathcal{C}$. 

Moreover, the interconnecting patterns pertaining to $\mathcal{E}^{(t_0+dt)}_\mathcal{C}$ compared to the ones found on $\mathcal{E}^{(t_0)}_\mathcal{C}$ would lack the near-consistency property to a great extent if $\mathrm{FCIX}_{t_0+dt}$ is relatively large. In view of this, making inference on dominance of relations existing among the assets at $t=t_0+dt$ would become a near impossible task due to the lack of proper edge weights $\mathcal{E}^{(t_0+dt)}_\mathcal{C}$ permitting existence of a quantitative comparison scheme among all the assets involved in $\mathcal{G}^{(t_0+dt)}_\mathcal{C}$. 

Insofar as the above-mentioned rationale is concerned, one possible definition for the market volatility that we propose is formulated in terms of the financial chaos index, which is presented below.

\begin{definition}[Market Volatility]
	Let the mutual information on the asset prices in the market be modeled using the tensorial formulation based on a given constitution criterion $\mathcal{C}$, and let $\mathcal{E}_\mathcal{C}^{(t_0)}$ be a random edge set related to the comparison multigraph $\mathcal{G}^{(t_0)}_\mathcal{C}$ at time $t=t_0$. Then, the market is said to be at a high-volatile (or high-chaos) regime at time $t=t_0$ if the random RPCM given by $ \mathbf{A}^{(t_0)}_\mathcal{C}$ associated to $\mathcal{E}_\mathcal{C}^{(t_0)}$ satisfies at least one of the following four properties:
	\begin{enumerate}
		\item (sensitivity) $ \exists \delta>0 , \forall \epsilon>0 : \lim\limits_{N\to \infty} \mathbb{P} \Big[  \dfrac{\left\lVert \mathrm{adj}(NI-\mathbf{A}^{(t_0)}_\mathcal{C}) \right\rVert_2}{N^{N-1}} < \delta \Big] < \epsilon  , $ \\ \\
		\item (consistency) $ \exists \delta>0 , \forall \epsilon>0 :  \lim\limits_{N\to \infty} \mathbb{P} \Bigg[  \dfrac{  \sum\limits_{i,j,l=1}^N  \mathbb{I}[a_{il}^{(t_0)}a_{lj}^{(t_0)}=a_{ij}^{(t_0)}]  }{N^3} < \delta \Bigg] < \epsilon  , $ \\ \\
		\item (discrepancy) $ \exists \delta>0, \forall \epsilon>0: \lim\limits_{N\to \infty} \mathbb{P} \Big[ \dfrac{\left\lVert \mathbf{A}^{(t_0)}-\mathbf{A}^{(t_0+dt)} \right\rVert_F}{\max\{\left\lVert \mathbf{A}^{(t_0)} \right\rVert_F,\left\lVert \mathbf{A}^{(t_0+dt)} \right\rVert_F\}}  < \delta \Big] < \epsilon  , $ \\ \\
		\item (homogeneity) \resizebox{0.78\hsize}{!} { $ \exists \delta>0 , \exists K>0, \forall \epsilon>0 :  \lim\limits_{N\to \infty} \mathbb{P} \Bigg[  \dfrac{   \sum\limits_{i,j=1}^N \big( \mathbb{I}[a_{ij}^{(t_0)}>K]+\mathbb{I}[a_{ij}^{(t_0)}<\frac{1}{K}] \big) }{2 N^2}   < \delta \Bigg] < \epsilon  . $} 
	\end{enumerate}
\end{definition}

According to the definition of the market volatility provided above, $\mathrm{FCIX}_{t=t_0}$ can be regarded as a suitable statistic to quantify the amount of market volatility at time $t_0\in \mathcal{T}$. In the following, we present several properties of our developed market volatility statistic:
\begin{enumerate}
	\item The statistic $\mathrm{FCIX}_{t=t_0}$ is a robust estimator of the market volatility at time $t_0\in \mathcal{T}$ as it relies on the largest eigenvalue of the RPCM at $t=t_0$, and hence, it possesses all the properties of the inconsistency measure $\psi(t=t_0)$ defined by \cref{Eq:inconsistency}.
	\item In contrast to various dispersion-based approaches, the financial chaos index does not require sample points (of size greater than one) realized from a particular asset's returns to quantify the market volatility.
	\item In contrast to various dispersion-based approaches, which utilize some value- or capitalization-weighting scheme to quantify the market volatility, the financial chaos index captures a different feature which pertains to quantifying the market volatility by measuring the inconsistency of mutual changes in assets' returns.
\end{enumerate}

%%%%%%%%%%%%%%%%%%%%%%%%%%%%%%%%%%%%%%%%%%%%%%%%%%%%%%%%%%%%%%%%%%%%%%%%%%%%%%%%%%%%%%%%%%%%%%%%%%%%%%%%%%%%%%%%%%%%%%%%%%%%%%%%%%%%%%%%

\subsection{Stock Market Segmentation Using FCIX}
\label{Section:2_4}
A brief inspection of \Cref{Fig_MFCIX} shows that the monthly $\mathrm{FCIX}_t$ data stream is made of consecutive regimes that are separated by abrupt changes, owing to the fact that the underlying model producing the data switches multiple times among various regimes. In such a situation, the realization of the monthly $\mathrm{FCIX}_t$ could be further characterized by resorting to \textit{retrospective change-point detection} methods. In this framework, the time series which are also referred to as signals, are segmented into several \textit{homogeneous} sub-signals.

The literature on the segmentation of time series by various change point detection methods is vast. A recent and thorough review of these methods is reported in \Citep{truong2018review}. In our case, we perform the change-point detection on a high-dimensional mapping of $\bar{\psi}_t$ which is implicitly defined by a \textit{kernel function}. This method is non-parametric and model-free, and it can be used to segment the time series without having any prior knowledge on the form of the underlying probability distribution that generates the data, see \Citep{arlot2019kernel,desobry2005online,harchaoui2007retrospective,harchaoui2009kernel} for more details. In short, $\bar{\psi}_t$ is first mapped onto a \textit{reproducing Hilbert space} $\mathcal{H}$ (i.e., a Hilbert space in which the point evaluation of functions is a certain continuous linear functional) for which the associated kernel function is denoted by 
$
k(\cdot,\cdot):\mathbb{R}_{\geq 0}^{\mathcal{T}}\times\mathbb{R}_{\geq 0}^{\mathcal{T}}\to \mathbb{R}  .
$
Further, the related mapping function $f:\mathbb{R}_{\geq 0}^{\mathcal{T}}\to\mathcal{H}$ is defined by
$
f(\bar{\psi}_t) = k(\bar{\psi}_t,\cdot) \in \mathcal{H}  ,
$
leading to the following definitions for the inner-product and norm
\begin{equation}
\langle f(\bar{\psi}_s)|f(\bar{\psi}_t)\rangle_{\mathcal{H}}  = k(\bar{\psi}_s,\bar{\psi}_t) , 
\end{equation}
and
\begin{equation}
\left\Vert f(\bar{\psi}_t) \right\Vert_{\mathcal{H}}^2 = k(\bar{\psi}_t,\bar{\psi}_t),
\end{equation}
respectively, for any samples $\bar{\psi}_s,\bar{\psi}_t \in \mathbb{R}_{\geq 0}^{\mathcal{T}}$. 

Next, assume that the kernel $k(\cdot,\cdot)$ is translation invariant, i.e.,
\begin{equation}
k(\bar{\psi}_s,\bar{\psi}_t) = \phi(\bar{\psi}_s-\bar{\psi}_t)  , \quad \forall s,t  ,
\end{equation}
where $\phi$ is a bounded continuous positive definite function on $\mathbb{R}$. Then, the mapping $f(\cdot)$ is shown to transform piecewise i.i.d. signals into piecewise constant signals within the feature space $\mathcal{H}$. That is, the signal becomes piecewise constant once it is mapped in the high-dimensional feature space $\mathcal{H}$ if it is comprised of independent random variables with piecewise constant distributions under the assumption that $k(\cdot,\cdot)$ is translation invariant, see \Citep{sriperumbudur2008injective} for more details.

Then, our change-point detection problem is aimed to detect mean-shifts in the embedded signal whose cost function is considered to be the \textit{average scatter measure} \Citep{harchaoui2007retrospective} given by
\begin{equation}
\label{Eq:Cost_Func}
\mathrm{cost}(\bar{\psi}_t) := \sum\limits_{t=a+1}^{b} \left\Vert f(\bar{\psi}_t) - \bar{\mu} \right\Vert_{\mathcal{H}}^2  ,
\end{equation}
where $\bar{\mu}\in\mathcal{H}$ denotes the empirical mean of the embedded sub-signal $\{f(\bar{\psi}_t)\}_{t=a+1}^{b}$. Note that explicit computation of the mapped data is not required since it can be shown with some effort that the kernel cost function given by \cref{Eq:Cost_Func} can be expressed as follows:
\begin{equation}
\label{Eq:Cost_Func_2}
\mathrm{cost}(\bar{\psi}_t) = \sum\limits_{t=a+1}^{b} k(\bar{\psi}_t,\bar{\psi}_t) - \frac{1}{b-a} \sum\limits_{s,t=a}^{b} k(\bar{\psi}_s,\bar{\psi}_t)  .
\end{equation} 

To compute the kernel cost function given by \cref{Eq:Cost_Func_2}, we employ the \textit{Gaussian kernel} (also known as the \textit{radial basis function}). This yields that
\begin{equation}
\label{Eq:Cost_Func_3}
\mathrm{cost}(\bar{\psi}_t) = (b-a) - \frac{1}{b-a} \sum\limits_{s,t=a}^{b} \exp(-\gamma \left\Vert \bar{\psi}_s-\bar{\psi}_t \right\Vert^2)  ,
\end{equation}
where $\gamma$ denotes a \textit{bandwidth} parameter.

For a fixed number of change points, say $K^\star$, the change-point detection problem translates to the following discrete optimization problem:
\begin{subequations}
	\label{Eq:Change_Point}
	\begin{align}
	\min &\,\,    \sum\limits_{k=0}^{K^\star} \mathrm{cost}(\{\bar{\psi}_t\}_{t=t_k}^{t_{k+1}}) \label{Eq:Change_Point_1}  \\
	\mathrm{s.t.} &\,\,   |\mathbf{t}| = K^\star   ,  \\ 
	&\,\,  0=t_0<t_1<\dots<t_{K^\star}<t_{K^\star+1}=|\mathcal{T}|  ,
	\end{align}
\end{subequations}
where $\mathbf{t} = \left[t_0,t_1,\cdots,t_{K^\star},t_{K^\star+1}\right]^\intercal$ denotes a vector containing the change points augmented by two dummy indexes $t_0:=0$ and $t_{K^\star+1}:=|\mathcal{T}|$. By noting the additive nature of the objective function \eqref{Eq:Change_Point_1}, this optimization problem can be solved recursively by the dynamic programming that relies on the observation that  \eqref{Eq:Change_Point} is equivalent to the following optimization problem:
\begin{subequations}
	\label{Eq:Change_Point_2}
	\begin{align}
	\min & \,\,  \mathrm{cost}(\{\bar{\psi}_t\}_{t=0}^{t^\star}) + \sum\limits_{k=0}^{K^\star-1} \mathrm{cost}(\{\bar{\psi}_t\}_{t=t_k}^{t_{k+1}}) \label{Eq:Change_Point_2_1}  \\
	\mathrm{s.t.} &\,\,   t^\star \leq |\mathcal{T}| - K^\star   ,  \\ 
	&\,\,  t^\star=t_0<t_1<\dots<t_{K^\star-1}<t_{K^\star}=|\mathcal{T}|  .
	\end{align}
\end{subequations}
This model delineates that once the optimal partition having $K^\star-1$ elements of all sub-signals is known, then the first change point of the optimal segmentation is expected to be relatively easily evaluated. 

\Cref{fig_Segmentation} depicts the results of solving model \eqref{Eq:Change_Point_2} for the realization $\bar{\psi}_t$ of a monthly $\mathrm{FCIX}_t$ during the time period from January 1990 to December 2019 using the Gaussian kernel cost function given by \cref{Eq:Cost_Func_3}. This figure indicates that the monthly $\mathrm{FCIX}_t$ switches cyclically among multiple regimes.

\begin{figure}
	\centering
%	\resizebox{1\textwidth}{0.4\textheight}{\input{./Figures/Fig_Segmentation.pgf}}
	\includegraphics[width=1\textwidth,height=0.4\textheight]{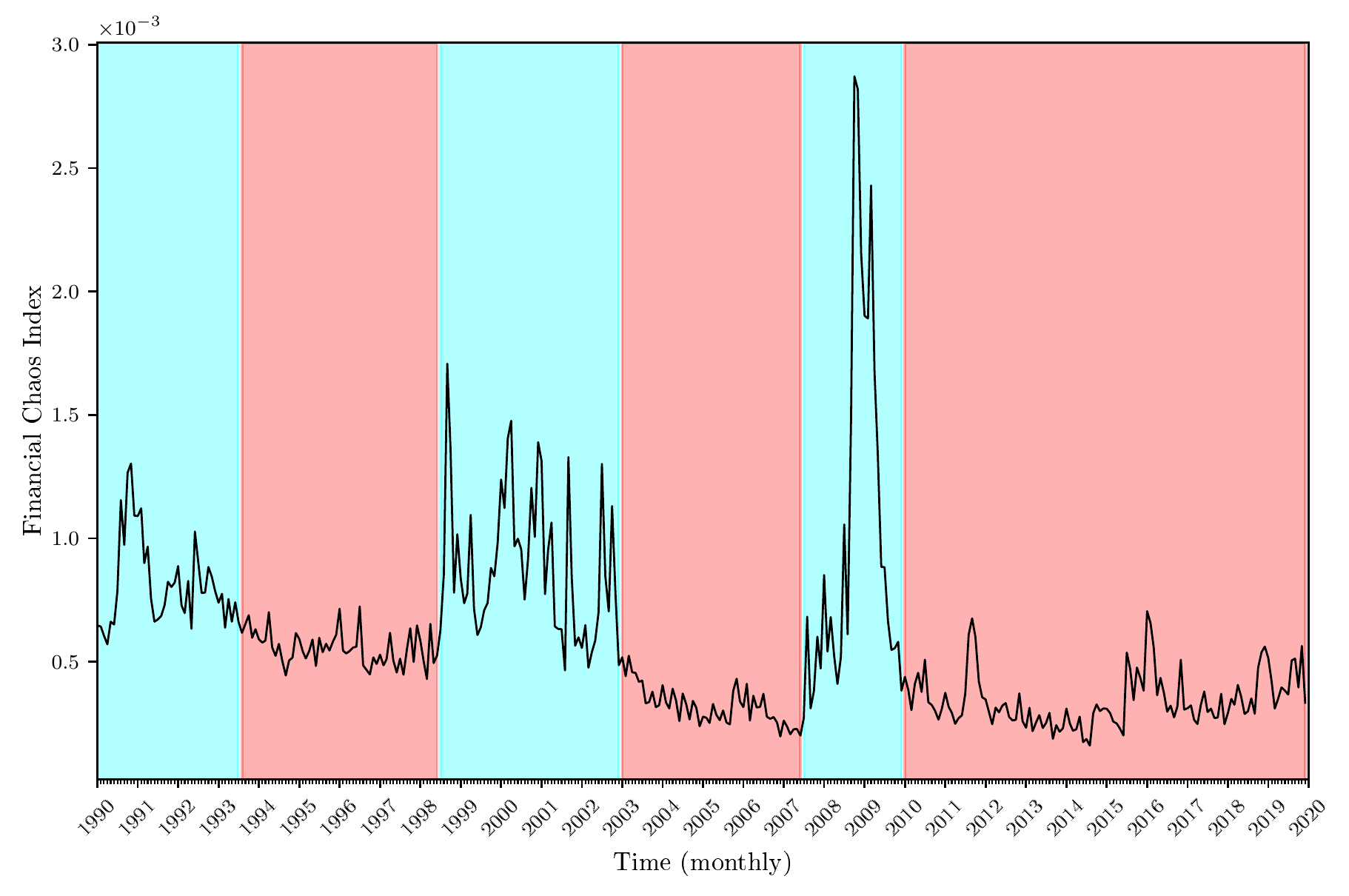}
	\caption[Schematic of the obtained segments for the market.]{Schematic of the obtained segments based on monthly $\mathrm{FCIX}_t$ during January 1990-December 2019.}
	\label{fig_Segmentation}
\end{figure}

\section{FCIX Versus VIX}
\label{Section:3}
In this section, we exploit the relationship between financial chaos index and \textit{CBOE volatility index} (VIX) which is a prominent measure for the market's expected volatility implied by S\&P$500$ options. Our goal would be to test whether $\mathrm{FCIX}_t$ and $\mathrm{VIX}_t$, also referred to as the \textit{fear index}, both reflect similar patterns of unusual behavior pertaining to the stock market. Then, we investigate their joint long-run behavior as well as their short-run kinematics, and also their possible causal relations. Thereafter, we borrow a variety of tools developed in information theory to model the dynamics of the realized and implied volatility by formulating a time-dependent dynamical model for the coupled $\mathrm{FCIX}_t$-$\mathrm{VIX}_t$ system.

%%%%%%%%%%%%%%%%%%%%%%%%%%%%%%%%%%%%%%%%%%%%%%%%%%%%%%%%%%%%%%%%$$$$$$$$$$$$$$$$$$$$$$$$$$$$$$$$$$$$$$$

%%%%%%%%%%%%%%%%%%%%%%%%%%%%%%%%%%%%%%%%%%%%%%%%%%%%%%%%%%%%%%%%%%%%%%%%%%%%%%%%%%%%%%%%%%%%%%%%%%%%%%%%

In \Cref{Subsection:2_3} we showed that the financial chaos index is a representative statistic for estimating the market volatility. However, it remains an important task to investigate the relationship between the realized volatility measured by $\mathrm{FCIX}_t$ and the implied volatility which indicates the forward-looking expectation of the volatility of the market which is measured by $\mathrm{VIX}_t$.
For this purpose, the use of cointegration techniques deems necessary for co-behavior analysis of two time series. By cointegrating these two time series, the possible short-run and long-run behavior existing between the series can potentially reveal themselves. More specifically, by analyzing the cointegrated relationship between $\mathrm{FCIX}_t$ and $\mathrm{VIX}_t$, a \textit{long-run equilibrium trajectory} can be defined, such that any departure from that path induces \textit{equilibrium correction} that move the coupled system back towards their stable trajectory. \Cref{fig_FCIX_VIX} depicts the plots of standardized daily $\mathrm{FCIX}_t$ and $\mathrm{VIX}_t$ during a time period from January 1990 to December 2019, based on which clear patterns of co-movement are observed between the series.

\begin{figure}
	\centering
%	\resizebox{1\textwidth}{0.4\textheight}{\input{./Figures/Chaos_Index/Fig_FCI_VIX.pgf}}
	\includegraphics[width=1\textwidth,height=0.4\textheight]{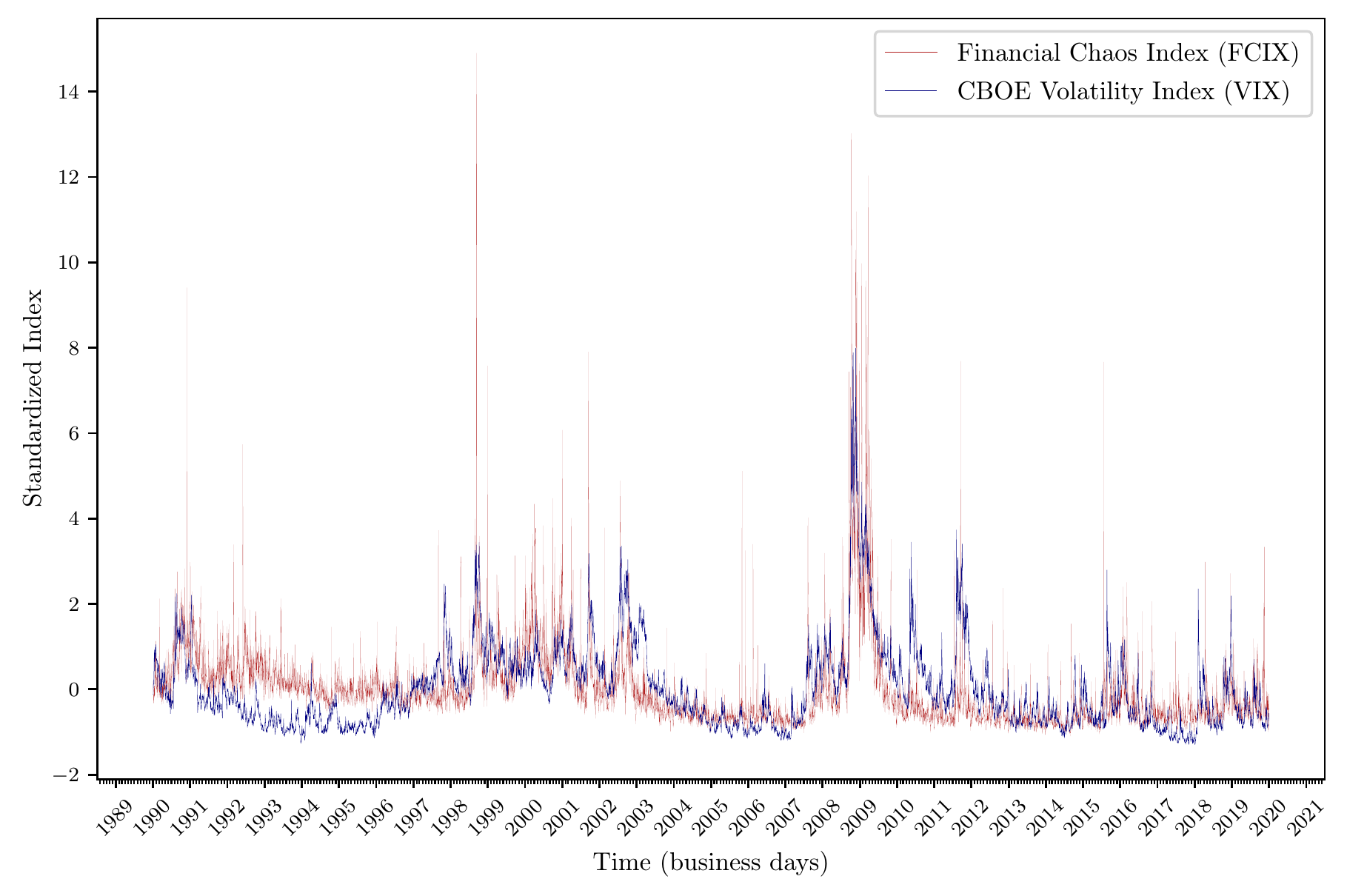}
	\caption[Plots of the daily $\mathrm{FCIX}_t$ vs $\mathrm{VIX}_t$.]{Plots of the standardized daily $\mathrm{FCIX}_t$ vs $\mathrm{VIX}_t$ during January 1990-December 2019.}	
	\label{fig_FCIX_VIX}
\end{figure}

%%%%%%%%%%%%%%%%%%%%%%%%%%%%%%%%%%%%%%%%%%%%%%%%%%%%%%%%%%%%%%%%%%%%%%%%%%%%%%%%%%%%%%%%%%%%%%%%%%%%%%%%

\subsection{Univariate Fractional Integration Analysis}
\label{Section:3_1}

As a preliminary step towards cointegration analysis of daily realizations of $\mathrm{FCIX}_t$ and $\mathrm{VIX}_t$, each time series is first examined individually by performing the \textit{Augmented Dickey-Fuller} (ADF) and \textit{Kwiatkowski, Phillips, Schmidt and Shin} (KPSS) tests for unit roots and stationarity, respectively. Based upon our experiments, both time series reject the null hypotheses of stationarity and presence of a unit root. As a general rule of thumb, if a time series rejects both the unit root and statianarity tests, it is commonly the indication of a situation in which the considered time series is fractionally integrated. That is, for a specific time series, say $X_t$, and for some value of the fractional integration order $d$, $X_t$ is said to be fractionally integrated of order $d$, denoted by $X_t \in I(d)$, in case $\Delta^d X_t \in I(0)$, i.e., $\Delta^d X_t$ is fractionally integrated of order zero. For our purpose, the parameter $d$ is considered to take its values from the interval $(-1.5,\infty)$. Further, the employed fractional difference operator $\Delta^d$ is defined by
\begin{equation}
\Delta^d X_t = \sum\limits_{n=0}^\infty \pi_n (-d) X_{t-n} ,
\end{equation}
where the coefficients $\pi_n (u)$ are obtained as follows:
\begin{equation}
\pi_n (u) = \frac{u (u+1) (u+2) \dots (u+n-1)}{n!} 
\end{equation}
which follows from the binomial expansion $(1-z)^{-u} =  \sum_{n=0}^\infty \pi_n (u) z_{n}$, e.g., see \Citep{johansen2014role, jensen2014fast} for more details on this expansion and efficient estimation of fractional differences.

Next, we plot the corresponding sample autocorrelation function and estimated power spectral density for each time series. \Cref{Fig:Autoorr_FCI,Fig:Autoorr_VIX} suggest that the autocorrelations for both time series decay hyperbolically which is the signature of fractional (long memory) time series, contrary to the geometric decay which pertains to short memory processes (the case of $d=0$). Furthermore, \Cref{Fig:Spectrum_FCI,Fig:Spectrum_VIX} reveal that the mass of spectrum is concentrated near the zero frequency for both time series, substantiating the consideration that the zero-frequency mass of each process is proportional to $\lambda^{-2d}$ ($\lambda$ denoting the frequency) for the respective values of the parameter $d$, a feature which is characteristics of the fractional time series. 

Subsequently, we compute the fractional integration order of each univariate series by resorting to the applications of \textit{extended local Whittle estimator} \Citep{abadir2007nonstationarity} which is consistent for $d\in (-1.5,\infty)$. By choosing the trimming parameter as $m=|\mathcal{T}|^{0.65}$ as recommended by the authors, estimates of the fractional integration orders are derived as being $\hat{d}=0.59$ and $\hat{d}=0.86$ for $\mathrm{FCIX}_t$ and $\mathrm{VIX}_t$, respectively. These results indicate that both time series are nonstationary having stationary increments, since their corresponding estimated fractional integration orders $\hat{d}\in [0.5,1.5)$.

\begin{figure}
	\centering
	\begin{minipage}[t]{0.45\linewidth}
%		\resizebox{0.9\textwidth}{0.4\textheight}{\input{./Figures/Chaos_Index/Fig_FCI_Autocorrelation.pgf}}
		\includegraphics[width=0.9\textwidth,height=0.4\textheight]{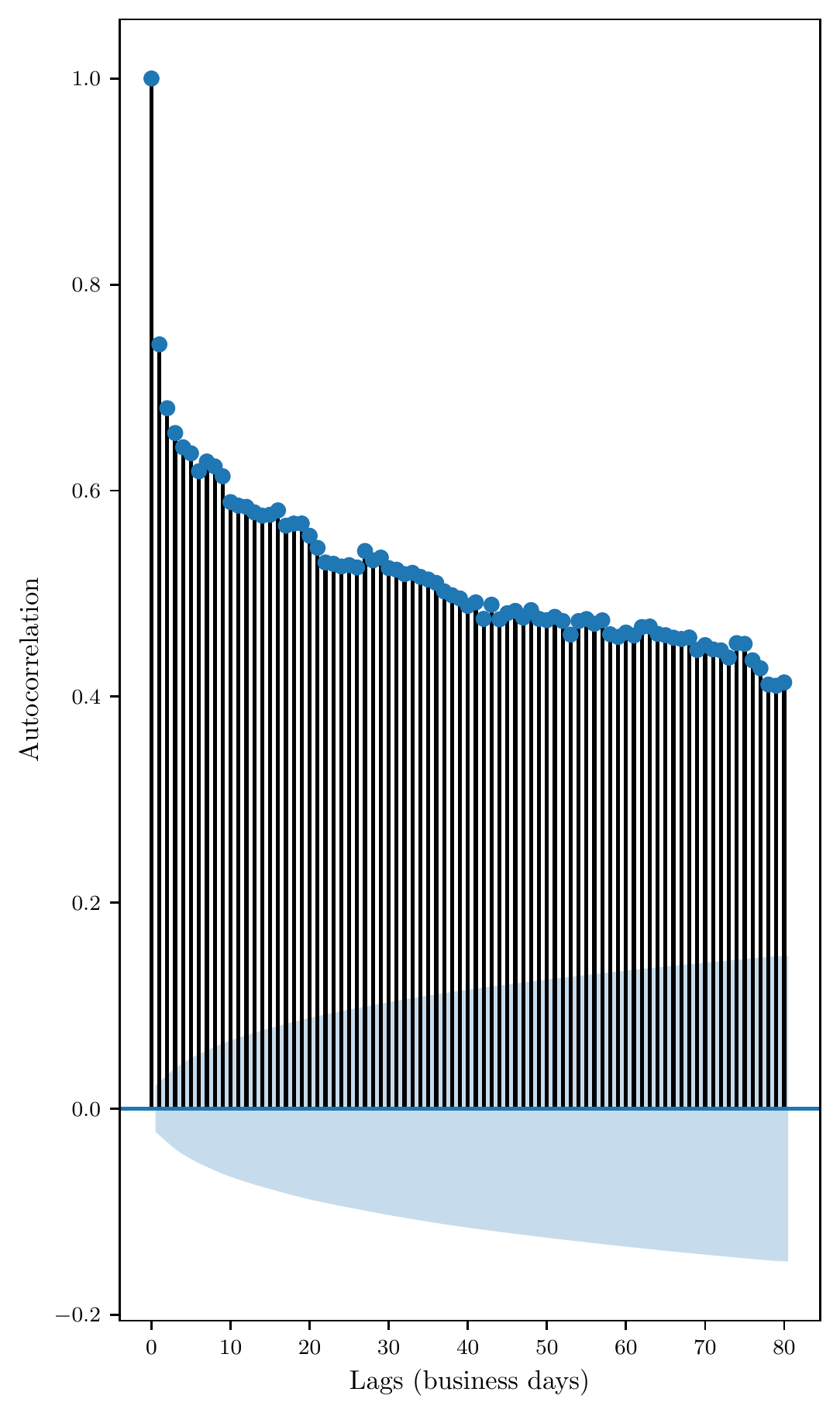}
		\caption{Autocorrelation function for $\mathrm{FCIX}_t$.}
		\label{Fig:Autoorr_FCI}
	\end{minipage}
	\begin{minipage}[t]{0.45\linewidth}
%		\resizebox{0.9\textwidth}{0.4\textheight}{\input{./Figures/Chaos_Index/Fig_VIX_Autocorrelation.pgf}}
	\includegraphics[width=0.9\textwidth,height=0.4\textheight]{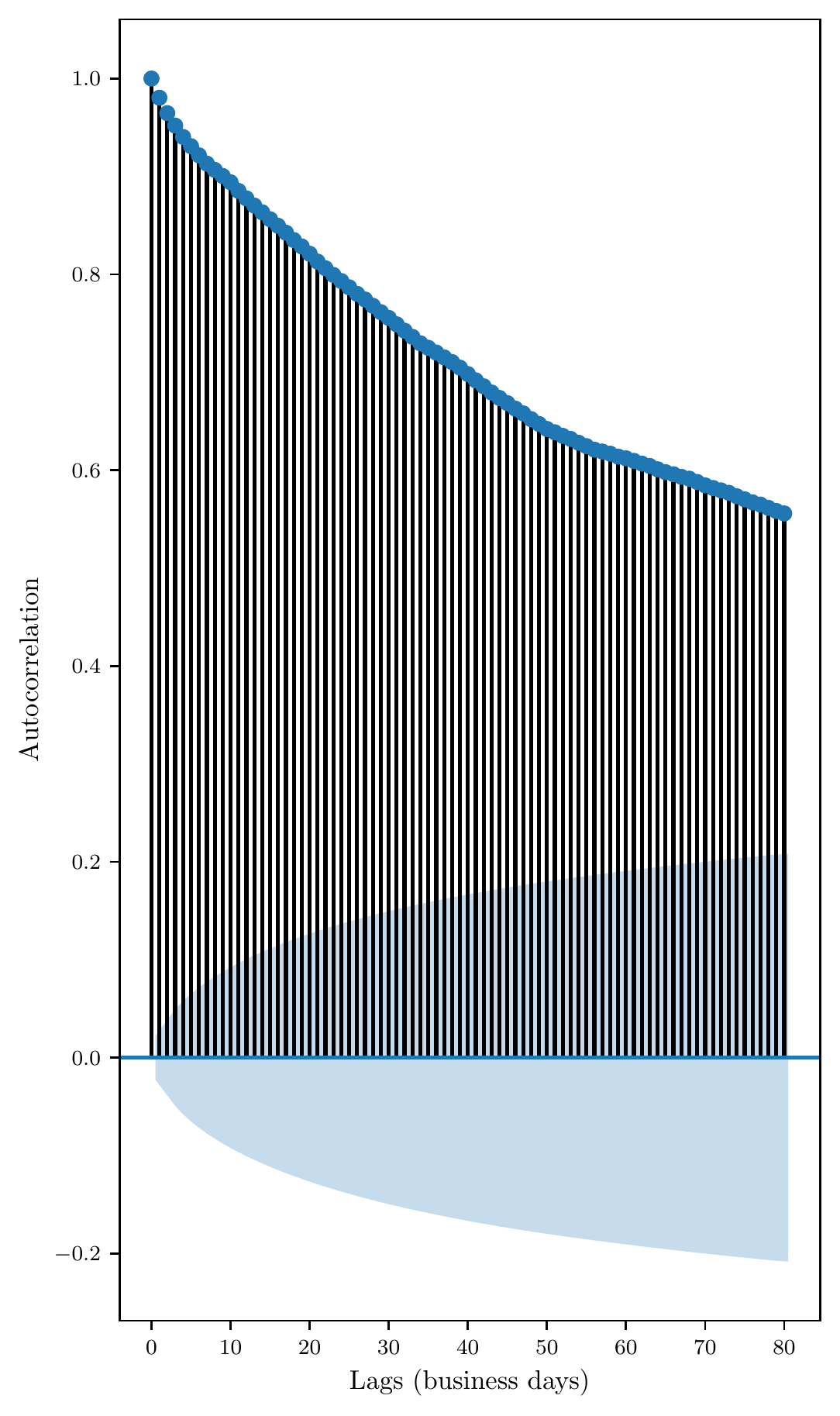}
		\caption{Autocorrelation function for $\mathrm{VIX}_t$.}
		\label{Fig:Autoorr_VIX}
	\end{minipage}
	\begin{minipage}[t]{0.45\linewidth}
%		\resizebox{0.9\textwidth}{0.4\textheight}{\input{./Figures/Chaos_Index/Fig_FCI_Spectrum.pgf}}
		\includegraphics[width=0.9\textwidth,height=0.4\textheight]{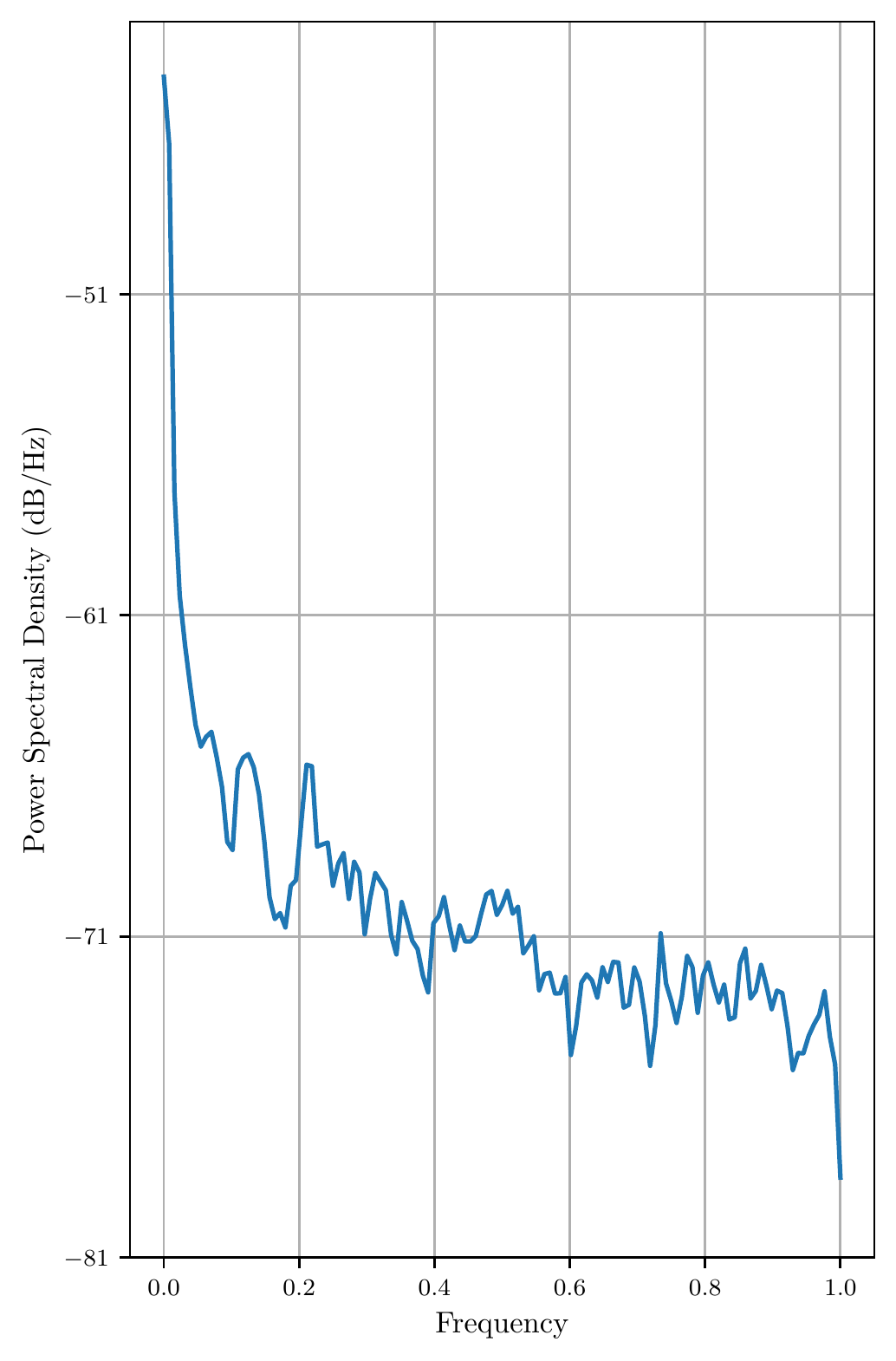}
		\caption{Power spectral density for $\mathrm{FCIX}_t$.}
		\label{Fig:Spectrum_FCI}
	\end{minipage}
	\begin{minipage}[t]{0.45\linewidth}
%		\resizebox{0.9\textwidth}{0.4\textheight}{\input{./Figures/Chaos_Index/Fig_VIX_Spectrum.pgf}}
		\includegraphics[width=0.9\textwidth,height=0.4\textheight]{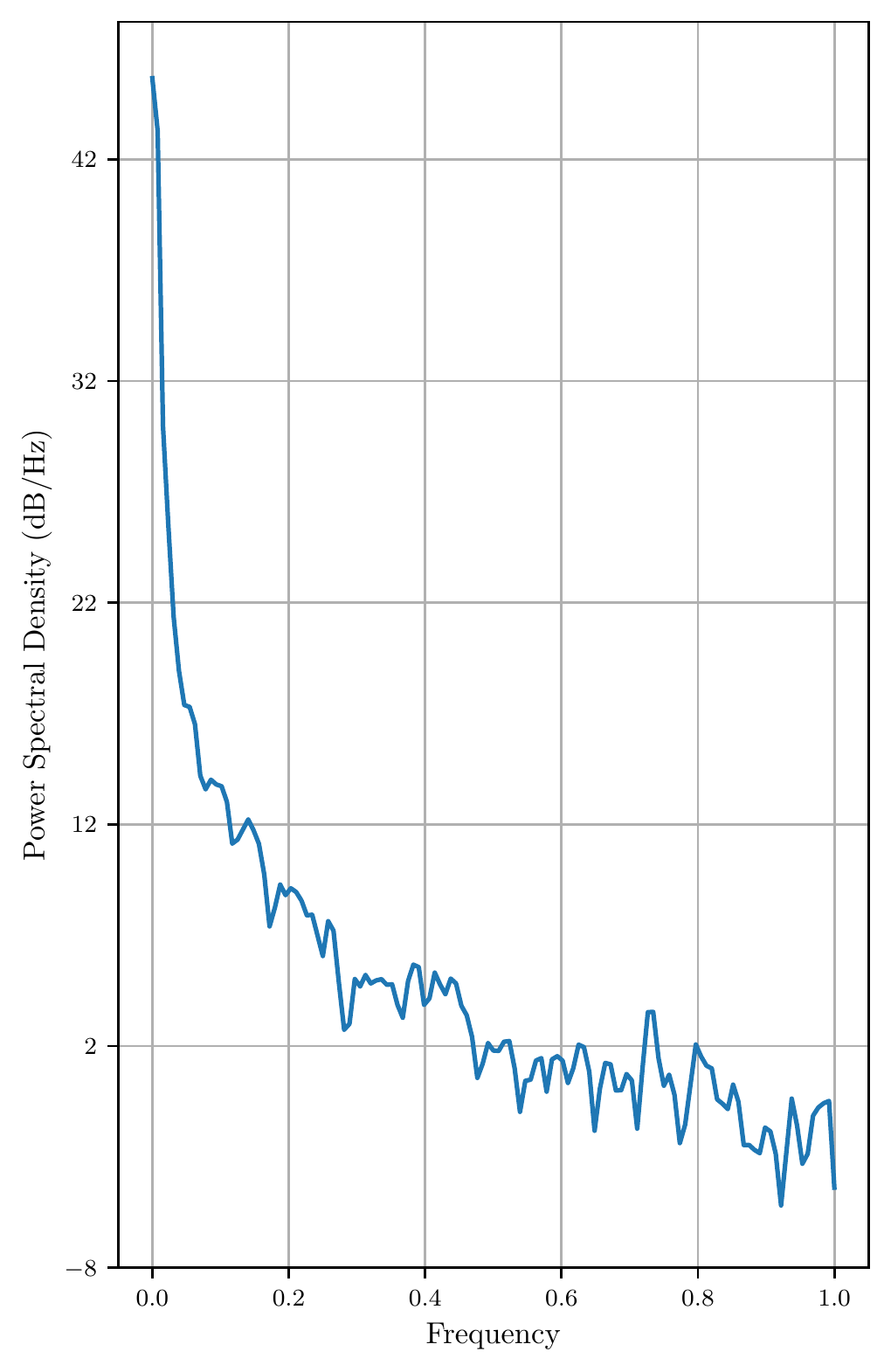}
		\caption{Power spectral density for $\mathrm{VIX}_t$.}
		\label{Fig:Spectrum_VIX}
	\end{minipage}
\end{figure}

%%%%%%%%%%%%%%%%%%%%%%%%%%%%%%%%%%%%%%%%%%%%%%%%%%%%%%%%%%%%%%%%%%%%%%%%%%%%%%%%%%%%%%%%%%%%%%%%%%%%%%%%

\subsection{Fractional Cointegration Analysis}
\label{Section:3_2}

To test for cointegration rank between $\mathrm{FCIX}_t$ and $\mathrm{VIX}_t$, we employ Johansen's procedure proposed in \Citep{johansen2008representation,} which was further developed and expanded in a series of scholastic works within the past decade, e.g., see \Citep{johansen2019nonstationary,johansen2012likelihood,johansen2018testing,dolatabadi2016fractionally} for more details. The initial steps of this procedure are to specify and estimate a \textit{fractionally cointegrated vector autoregressive} (FCVAR) model of order $p$ for $\mathbf{X}_t = ( \mathrm{FCIX}_t, \mathrm{VIX}_t )^\intercal$ and then determine the number of cointegrating vectors by performing a likelihood ratio test for the rank of the matrix $\mathbf{\Pi}=\boldsymbol{\alpha} \boldsymbol{\beta}^\intercal$, which is also called the \textit{long-run impact matrix}. This is followed by estimating the resulting \textit{vector error correction model} (VECM) by the method of maximum likelihood. More specifically, the level FCVAR($p$) model in error correction form is given by
\begin{equation}
\label{EQ:VECM_P}
\Delta^d\mathbf{X}_t = \boldsymbol{\alpha} \Delta^{d-b}L_b(\boldsymbol{\beta}^\intercal \mathbf{X}_t + \boldsymbol{\rho}^\intercal) +  \sum_{i=1}^{p} \mathbf{\Gamma}_i \Delta^d L_b^i \mathbf{X}_{t} + \boldsymbol{\zeta} + \boldsymbol{\epsilon}_t  ,
\end{equation}
where $L_b=1-\Delta^d$ is a \textit{fractional lag operator} and $b$ denotes the \textit{degree} of fractional cointegration. The innovations $\{\boldsymbol{\epsilon}_t\}$ follow $N(0,\boldsymbol{\Sigma})$, which are further assumed to be uncorrelated over time. Also, $\boldsymbol{\alpha}$ and $\boldsymbol{\beta}$ are $2\times r$ matrices with $0\leq r \leq 2$. These matrices play important roles in separating the short-run and long-run relationships between $\mathrm{FCIX}_t$ and $\mathrm{VIX}_t$. For instance, $\boldsymbol{\alpha}$ determines the corresponding \textit{error-correction speeds} which distributes the influence of $\mathrm{FCIX}_t$ and $\mathrm{VIX}_t$ onto the temporal behavior of $\Delta^d\mathbf{X}_t$, whereas columns of $\boldsymbol{\beta}$ provide the basis vectors for the space of fractionally cointegrating relations. That is, $\boldsymbol{\beta}^\intercal \mathbf{X}_{t}$ can be considered as the error-correction term, reflecting long-run equilibrium relations between $\mathrm{FCIX}_t$ and $\mathrm{VIX}_t$. The matrices $\mathbf{\Gamma}_i$ are further referred to as \textit{short-run impact matrices} which capture the short-run dynamics of the variables. Besides, the parameter $\boldsymbol{\rho}$ represents a constant term in the model designating the mean level of the long-run equilibria such that $\mathbb{E}[\boldsymbol{\beta}^\intercal \mathbf{X}_t] + \boldsymbol{\rho}^\intercal = 0$. An additional constant term giving rise to a deterministic trend in the levels of $\mathrm{FCIX}_t$ and $\mathrm{VIX}_t$ is further captured by the parameter $\boldsymbol{\zeta}$. Note that in the case of $d=1$ the trend would be linear, and if in addition $b=d=1$, then $\boldsymbol{\rho}$ will be absorbed in $\boldsymbol{\zeta}$.

%One practical reformulation of the model represented in \cref{EQ:VECM_P} to which we adhere in this work, is expressed as follows:
%\begin{equation}
%\label{EQ:VECM_Mu}
%\Delta^d(\mathbf{X}_t-\boldsymbol{\mu}) = \boldsymbol{\alpha} \boldsymbol{\beta}^\intercal \Delta^{d-b}L_b( \mathbf{X}_t - \boldsymbol{\mu}) +  \sum_{i=1}^{p} \mathbf{\Gamma}_i \Delta^d L_b^i( \mathbf{X}_t - \boldsymbol{\mu}) + \boldsymbol{\epsilon}_t  ,
%\end{equation}
%where $\boldsymbol{\mu}$ denotes a level parameter, accommodating a non-zero initial point for the process. By including the deterministic term $\boldsymbol{\mu}$, each of the series are shifted by a constant which will avoid encountering the bias which comes about due to pre-sample observations of the process, e.g., see \Citep{johansen2018testing} for more details. We note that $\boldsymbol{\beta}^\intercal \boldsymbol{\mu} = -\boldsymbol{\rho}^\intercal$ designates the mean of the stationary cointegrating relations.

Subsequently, the fractional cointegration rank $r$ can be determined by examining the following nested hypotheses
\begin{equation*}
\mathrm{H}_0(r_0) : r = r_0  ,
\end{equation*}
versus their alternatives
\begin{equation*}
\mathrm{H}_1(r_0) : r = 2  .
\end{equation*}
We may then formulate the error-correction fractional cointegrating null hypotheses and their alternatives as follows:
\begin{equation*}
\begin{cases}
\mathrm{H}_{0,1}^{\mathrm{EC}} : r = 0  , \\
\mathrm{H}_{0,2}^{\mathrm{EC}} : r = 1  ,
\end{cases}
\hspace{2cm}
\begin{cases}
\mathrm{H}_{1,1}^{\mathrm{EC}} : r = 2  , \\
\mathrm{H}_{1,2}^{\mathrm{EC}} : r = 2  .
\end{cases}
\end{equation*}
These hypothesis could be tested in the following order $\mathrm{H}_{0,1}$ and $\mathrm{H}_{0,2}$ such that $\mathrm{H}_{0,2}$ can only be rejected if $\mathrm{H}_{0,1}$ is rejected, i.e., the testing procedure gets terminated by returning the result $r=0$ if one fails to reject $\mathrm{H}_{0,1}$, otherwise $\mathrm{H}_{0,2}$ is tested to determine whether $r=1$ or not. Note that a likelihood ratio (LR) test statistic called \textit{trace statistic} is employed to examine these hypotheses whose details are provided in \Citep{johansen2018testing}. 

\Cref{Coint_Trace} reports the results for testing $\mathrm{H}_{0,1}$ and $\mathrm{H}_{0,2}$ obtained using the computer program provided in \Citep{nielsen2016matlab}. As per this table, $\mathrm{H}_{0,1}$ gets rejected at $1\%$ significance level, which implies that $r>0$. On the other hand, there is insufficient evidence to reject  $\mathrm{H}_{0,2}$. Consequently, we may conclude that $\mathbf{X}_t$ is cointegrated with one cointegrating vector (i.e., $\mathrm{rank}(\mathbf{\Pi}) = 1$), which in turn implies that there exist a common fractionally integrated trend driving both $\mathrm{FCIX}_t$ and $\mathrm{VIX}_t$. It is worth mentioning that this common trend is more likely to be the so-called \textit{implicit common efficient price} which permanently inflicts movements in both time series by following the new market information. Further, note that the optimal lag value in our experiments was identified to be $p=7$ by fitting FCVAR models to the data and evaluating the performance of various model selection criteria such as AIC, BIC and HQC.

\begin{table}
	\centering
	\caption{Results of likelihood ratio test for cointegration rank.}
	\label{Coint_Trace}
	\begin{tabular}{cccc} 
		\toprule
		Rank & Log-likelihood & LR statistic & $p$-value  \\ 
		\hline
		0&   37657.897                &   16.536           &  0.002        \\
		1&   37666.060                &   0.211           &    0.646      \\
		2&   37666.165               &     -         &    -      \\
		\bottomrule
	\end{tabular}
\end{table}

Furthermore, on the basis of the experiments which we have conducted, parameters of the model given in \cref{EQ:VECM_P} were estimated as follows (standard errors in parentheses):
\begin{equation*}
\hat{d} = 0.542 \, (0.080), \hat{b} = 0.212 \, (0.115), \, \hat{\boldsymbol{\rho}}=0.598,
\end{equation*}
\begin{equation*}
\hat{\boldsymbol{\alpha}}=
\begin{bmatrix}
0.482 & 0.117
\end{bmatrix}^\intercal,
\end{equation*}
\begin{equation*}
\hat{\boldsymbol{\beta}}=
\begin{bmatrix}
1.000 & -1.199
\end{bmatrix}^\intercal,
\end{equation*}
\begin{equation*}
\hat{\boldsymbol{\zeta}}=
\begin{bmatrix}
-0.016 & -0.004
\end{bmatrix}^\intercal,
\end{equation*}
and
\begin{equation*}
\label{Eq:Impact_Matrix}
\reallywidehat{\mathbf{\Pi}} = 
\begin{bmatrix}
0.482 & -0.577 \\ 0.117 & -0.141
\end{bmatrix}	 .
\end{equation*}

In turn, the obtained estimates presented above enable us to characterize the long-run relationship between standardized values of $\mathrm{FCIX}_t$ and $\mathrm{VIX}_t$. We illustrate this relationship schematically in \Cref{fig:Equilibrium}, where the horizontal axis represents a stable relation between variables. This figure indicates to a presence of a long-run common movement between $\mathrm{FCIX}_t$ and $\mathrm{VIX}_t$ along the considered time horizon. Note that $\mathrm{FCIX}_t$ and $\mathrm{VIX}_t$ are tied together through cointegration, and they both respond to disturbances which make adjustments towards the long-run equilibrium relationship. In this respect, \Cref{fig:Equilibrium} reveals that there exist several departures of large magnitude from the equilibrium behavior of the two time series. Such fluctuations are observed to occur mainly during the time periods when the market has undergone severe crises, some of which for instance were listed on previous pages to have happened within the time frame $2008-2009$.

\begin{figure}
	\centering
%	\resizebox{1\textwidth}{0.35\textheight}{\input{./Figures/Chaos_Index/Fig_FCI_VIX_Long_Equilibrium.pgf}}
	\includegraphics[width=1\textwidth,height=0.35\textheight]{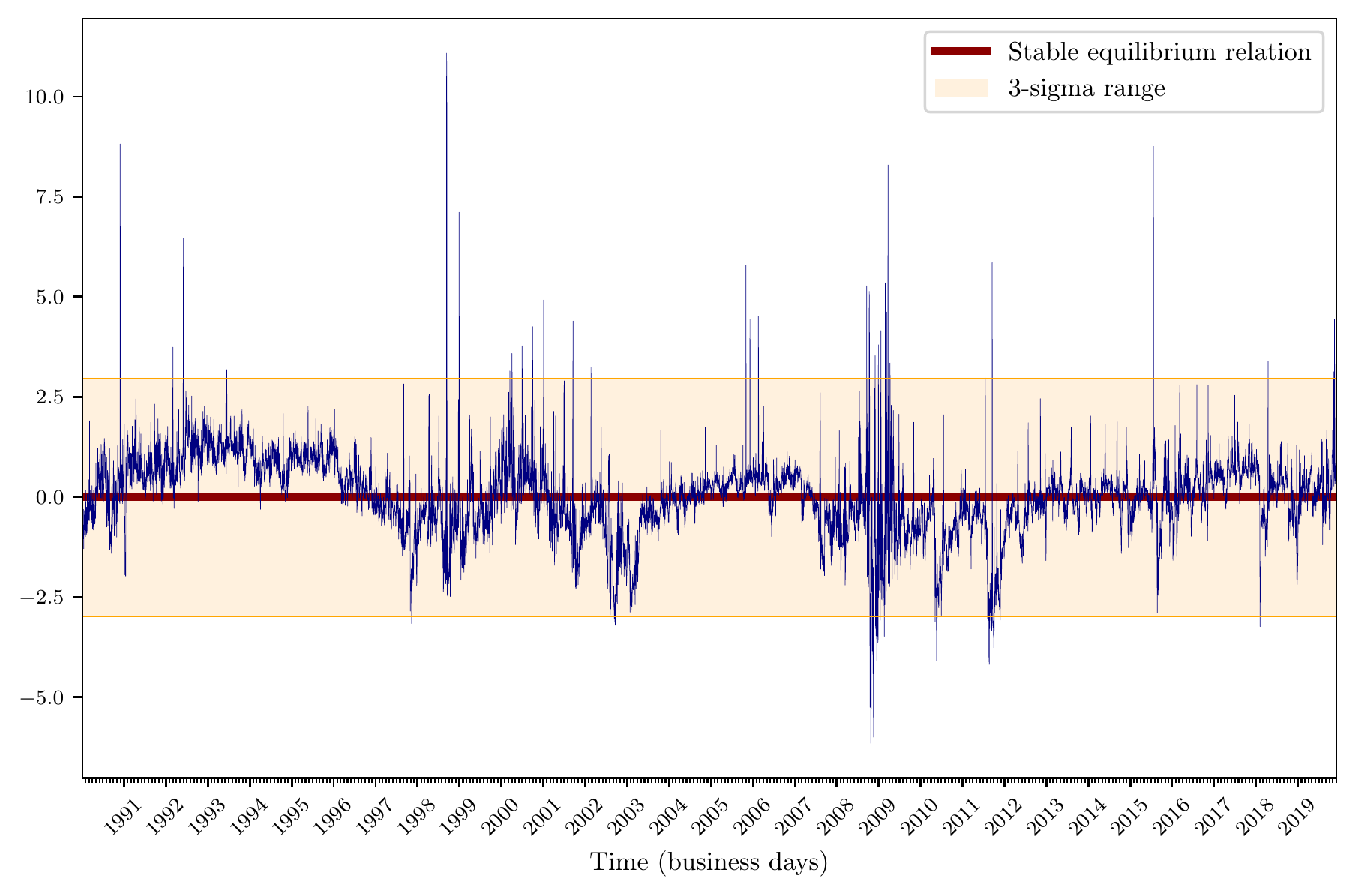}
	\caption[Plot of the long-run equilibrium for daily $\mathrm{FCIX}_t$-$\mathrm{VIX}_t$.]{Plot of the long-run equilibrium between standardized daily $\mathrm{FCIX}_t$ and $\mathrm{VIX}_t$ during January 1990-December 2019.}	
	\label{fig:Equilibrium}
\end{figure}

In contrast to the analysis of the long-run equilibrium, identification of the short-run relations appear to be a more difficult task since there usually are quite a few matrices of coefficients whose interpretations are often not trivial. To surmount this obstacle, we incorporate the notions of the \textit{impulse-response function} (IRF) and the \textit{cross-correlation function} (XCF) along with the notion of approximate entropy within our framework to study the existing short-run relations. As depicted in \Cref{Fig:Cross_Corr}, the time delay at which two signals are aligned best corresponds to $p=0$, where the XCF is roughly $60\%$. However, it is observed that the amount of XCF remains relatively high even when the time lag parameter is considered to be as large as $p=\pm 80$ business days, implying that the impacts of any change in either of the time series could persist in the other series for several time periods ahead.

In order to investigate the causal relations which might (or might not) exist within the coupled economical system $\mathrm{FCIX}_t-\mathrm{VIX}_t$ and track an impact that either time series makes on the other, we confine ourselves to the framework of \textit{empirical causal analysis} which emphasizes the \textit{orthogonalized impulse-response functions} to conduct the causal analysis in between these two time series. Basically, an impulse-response function describes the evolution of the reaction from one time series along a specified time horizon after a shock is inflicted on the other time series at a given moment. Then, in order to estimate the effect of a one-unit shock in $\boldsymbol{\epsilon}_t$ to variables of the system, we decompose the variance-covariance matrix into a lower triangular matrix using the \textit{Cholesky decomposition}. Let $\mathbf{P}$ denote the Cholesky decomposition of $\boldsymbol{\Sigma}$ such that $\mathbf{P}\mathbf{P}^\intercal=\boldsymbol{\Sigma}$ and $\mathbf{P}$ is a lower triangular matrix, and assume that $\mathbf{u}_t$ is such that $\mathbf{u}_t=\mathbf{P}^{-1}\boldsymbol{\epsilon}_t$. Our experiments, led to the following estimate for $\mathbf{P}$:
\begin{equation}
\widehat{\mathbf{P}} = 
\begin{bmatrix}
0.0003 & 0.0000 \\ 0.1336 & 1.4898
\end{bmatrix}	 \cdot
\end{equation}
We note that the orthogonalized impulse response function would reflect a one standard-deviation impulse to $\mathbf{u}_t$ since
\begin{equation}
\mathbb{E}[\mathbf{u}_t \mathbf{u}_t^\intercal] = \mathbf{P}^{-1} \mathbb{E}[\boldsymbol{\epsilon}_t \boldsymbol{\epsilon}_t^\intercal] (\mathbf{P}^\intercal)^{-1} = \mathbf{I}_2  .
\end{equation}

\Cref{IRF_FCIX_VIX,IRF_VIX_FCIX} depict the effect of the impact of a one-unit shock in $\mathrm{FCIX}_t$ on $\mathrm{VIX}_t$ for $80$ business days ahead and vice versa, respectively. In order to characterize the regularity and information content of the orthogonalized IRFs, the approximate entropy of the signals are further computed, yielding $0.098$ for $\mathrm{FCIX}_t \leadsto \mathrm{VIX}_t$ and $0.160$ for $\mathrm{VIX}_t \leadsto \mathrm{FCIX}_t $. These quantities imply that $\mathrm{FCIX}_t$ is more sensitive to changes in $\mathrm{VIX}_t$ rather than the other way around, as the orthogonalized IRF for $\mathrm{VIX}_t \leadsto \mathrm{FCIX}_t$ attains a higher value of approximate entropy, and hence, it has a higher degree of irregularity and unpredictability. It is also worth to mention the synchronicity which is evident like that when both orthogonalized IRF graphs reach their maximuma at day $10$ after the shocks were being imposed. Besides, a magnitude of response for the case of $\mathrm{VIX}_t \leadsto \mathrm{FCIX}_t $ is considerably larger than its counterpart $\mathrm{FCIX}_t \leadsto \mathrm{VIX}_t$, which stipulates that a one-unit shock in $\mathrm{VIX}_t$ causes a more intense response pertaining to $\mathrm{FCIX}_t$ as compared to the case where the shock is imposed on $\mathrm{FCIX}_t$.

One rationale that partially accounts for the above-mentioned assessments is rooted in the fact that the financial chaos index is a robust estimator of the mutual price changes of the market assets. That is, any substantial marginal increase in the value of $\mathrm{FCIX}_t$ would demand majority of the underlying assets (S\&P$500$) to depart significantly from their steady-state price levels. As a consequence, any abrupt change in the value of $\mathrm{FCIX}_t$ due to a shock could potentially be an important indicator of some structural changes in the underlying assets. In contrast to abrupt changes in $\mathrm{FCIX}_t$, imposing a one-unit shock in $\mathrm{FCIX}_t$ is less likely to be a cause of significant structural changes in the asset prices. Thus, the response of  such a shock cannot have a large magnitude since $\mathrm{VIX}_t$ would respond sharply to a shock in $\mathrm{FCIX}_t$ only if the underlying equities (S\&P$500$) utilized in its computation undergo major changes. The foregoing statement accounts for the relatively lower magnitude of $\mathrm{VIX}_t$ response to a shock in $\mathrm{FCIX}_t$, as depicted in \Cref{IRF_FCIX_VIX}. This figure further discloses that an increase in the $\mathrm{FCIX}_t$ values causes positive increments in $\mathrm{VIX}_t$ response with the positive effect peak to occur within the first few days, which is then followed by an oscillation until it meets the highest-magnitude response. Thereafter, the response of $\mathrm{VIX}_t$ to a shock in $\mathrm{FCIX}_t$ remains positive for the remainder of the time horizon.

Conversely, the response of $\mathrm{FCIX}_t$ to a shock in $\mathrm{VIX}_t$ has a slightly different behavior than that explained above. As depicted in \Cref{IRF_VIX_FCIX}, $\mathrm{VIX}_t$ has a positive \textit{contemporaneous effect} on $\mathrm{FCIX}_t$ and this positive response persists during the whole time horizon. However, it is relevant to stress the presence of a steep decline in response of $\mathrm{FCIX}_t$ which occurs immediately after the mentioned contemporaneous effect at time instant $t_0$ (the "double dip" observed within $[t_0,t_0 + 5]$). In order to interpret this empirical observation, we find it useful to recall that the lower values of the financial chaos index are linked to a market in which the \textit{collective judgment} of its participants remain consistent over time (here, we understand consistency in the same sense as that described under the \blue{PCJ} \ref{PCJ_Axiom}). Thus, a positive increase in $\mathrm{FCIX}_t$ caused by the contemporaneous effect could potentially be understood as a deviation from the \textit{collective consistency condition}. On the other hand, the response from $\mathrm{FCIX}_t$ decreases immediately after a day is elapsed since the collective judgment of the participants in the market becomes more aware of the inconsistencies that have occurred due to an increased amount of the expected volatility. Thus, the realized market volatility reacts to this phenomenon by correcting its value. However, the response from $\mathrm{FCIX}_t$ spikes once again after a few days, albeit this time on the reason that the implied volatility becomes realized in the market, which further makes the task of casting a series of consistent judgments less probable for the agents, which in turn gives rise to increased values of the $\mathrm{FCIX}_t$ response.

To conclude our discussion, it can be said that $\mathrm{VIX}_t$ has a stronger causal effect on $\mathrm{FCIX}_t$ as compared to the opposite. Hence, the implied volatility of the market plays a key role in characterization of the market's realized volatility, further implying that a set of feedback and feedforward causal relations govern the interactions between $\mathrm{FCIX}_t$ and $\mathrm{VIX}_t$.

\begin{figure}
	\centering
%	\resizebox{0.9\textwidth}{0.3\textheight}{\input{./Figures/Fig_Cross_Corr_FCE_VIX.pgf}}
	\includegraphics[width=0.9\textwidth,height=0.3\textheight]{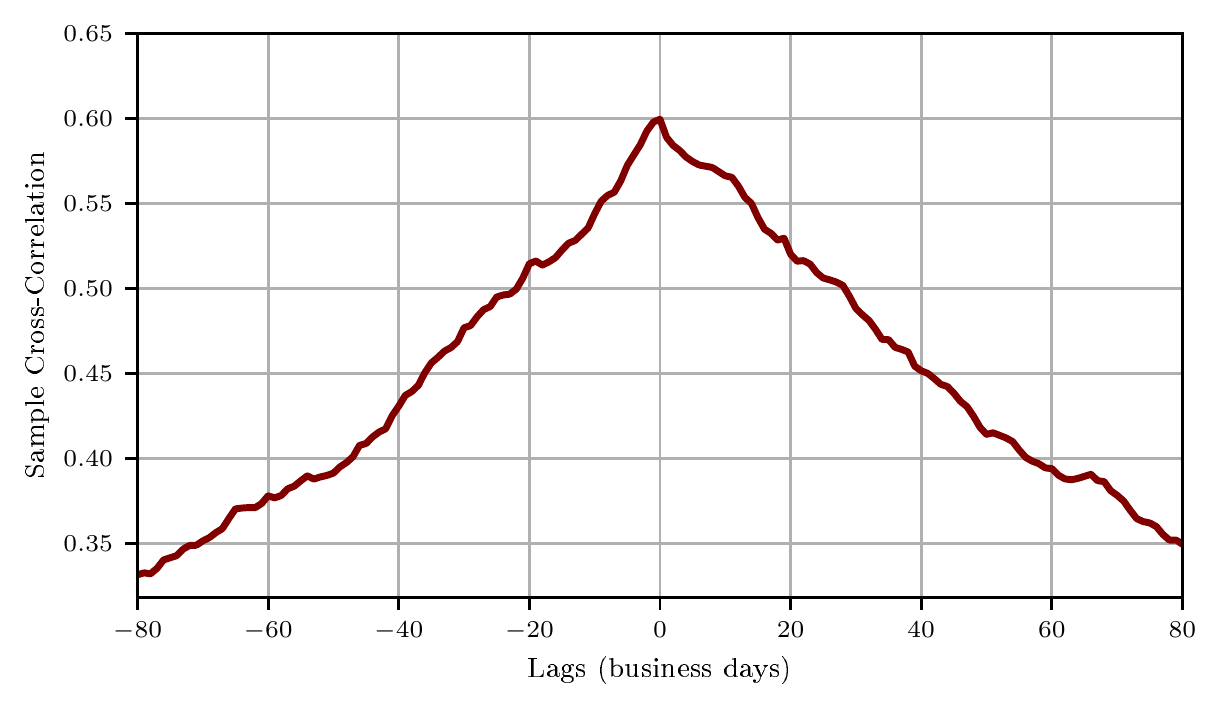}
	\caption[Plot of the sample cross-correlation for daily $\mathrm{FCIX}_t$-$\mathrm{VIX}_t$.]{Plot of the sample cross-correlation between daily $\mathrm{FCIX}_t$ and $\mathrm{VIX}_t$ during January 1990-December 2019.}
	\label{Fig:Cross_Corr}
	\vspace{0.5cm}
%	\resizebox{0.9\textwidth}{0.3\textheight}{\input{./Figures/Entropy/Fig_IRF_FCI_VIX.pgf}}
\includegraphics[width=0.9\textwidth,height=0.3\textheight]{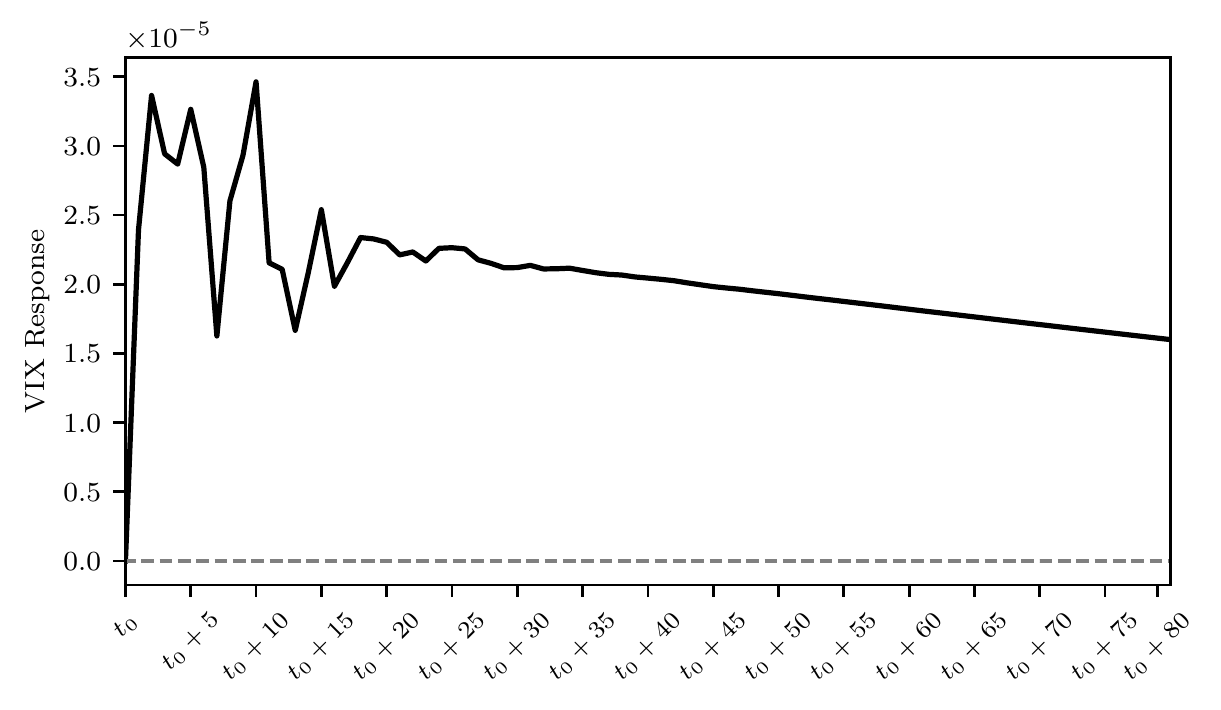}
	\caption[Plot of the orthogonalized IRF for $\mathrm{FCIX}_t \leadsto \mathrm{VIX}_t$.]{Plot of the orthogonalized IRF for $\mathrm{FCIX}_t \leadsto \mathrm{VIX}_t$ during January 1990-December 2019, $\mathrm{ApEn}=0.098$.}
	\label{IRF_FCIX_VIX}
	\vspace{0.5cm}
%	\resizebox{0.9\textwidth}{0.3\textheight}{\input{./Figures/Entropy/Fig_IRF_VIX_FCI.pgf}}
\includegraphics[width=0.9\textwidth,height=0.3\textheight]{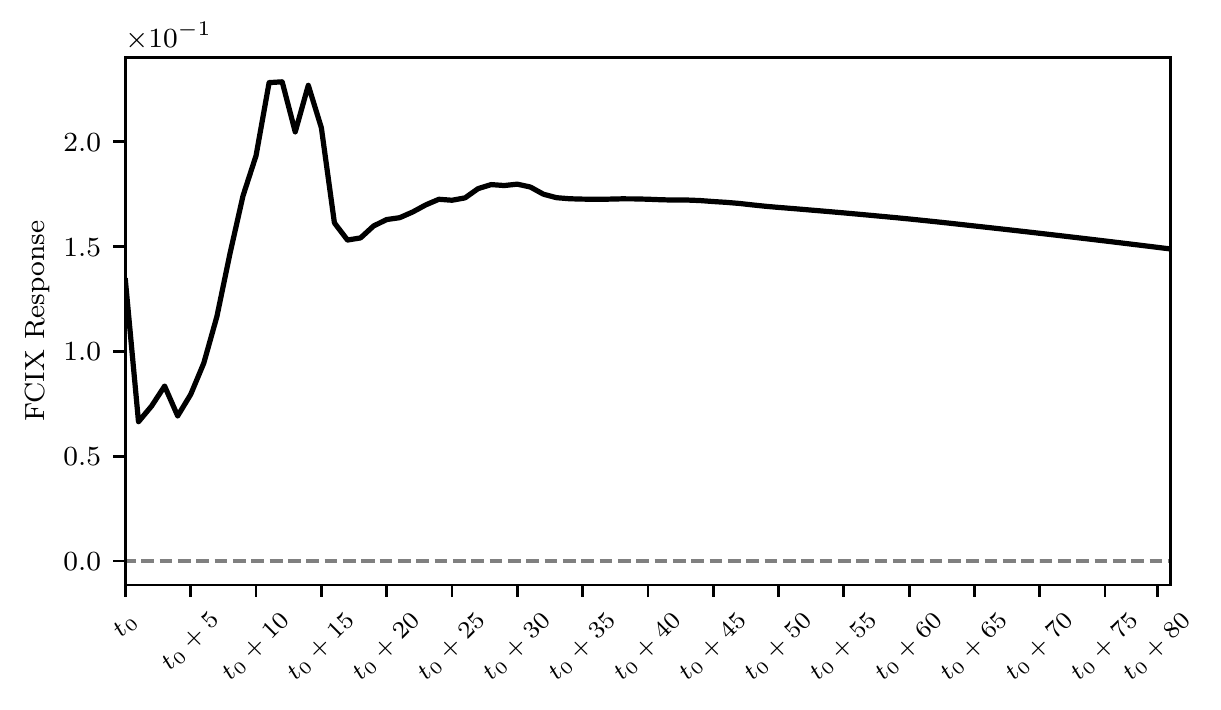}
	\caption[Plot of the orthogonalized IRF for $\mathrm{VIX}_t \leadsto \mathrm{FCIX}_t $.]{Plot of the orthogonalized IRF for $\mathrm{VIX}_t \leadsto \mathrm{FCIX}_t $ during January 1990-December 2019, $\mathrm{ApEn}=0.160$.}
	\label{IRF_VIX_FCIX}
\end{figure}

%%%%%%%%%%%%%%%%%%%%%%%%%%%%%%%%%%%%%%%%%%%%%%%%%%%%%%%%%%%%%%%%%%%%%%%%%%%%%%%%%%%%%%%%%%%%%%%%%%%%%%%%%

\subsection{Information Dynamics}
\label{Section:3_3}
In the previous subsection, we investigated the short- as well as long-run co-behavior of $\mathrm{FCIX}_t$ and $\mathrm{VIX}_t$ and showed the existence of a potential bi-directional causal relation between these two time series. In what follows, we consider an additional aspect of the dynamical relation existing between the series which pertains to the scheme by which the information content flows from one time series to another and their possible interactions. For this purpose, we first assume that the coupled system $\mathrm{FCIX}_t$-$\mathrm{VIX}_t$ is a given isolated system. We will then introduce and utilize the notions of \textit{transfer entropy} and \textit{self-entropy} to quantify the dynamics of the information flow within the coupled system by taking into account the temporal evolution of the existing information dynamics.

Let $V$ and $Z$ denote two dynamical processes which further represent a time-evolutionary coupled dynamical system, and let us assume that $Z$ is assigned to be the target process. Then, one can employ some standard information-theoretic measures to study the information relevant to $Z$ \Citep{cover2012elements,faes2014conditional}. A pivotal measure of such kind is the well-known \textit{Shannon entropy} which quantifies the amount of information carried by $Z$ in terms of its average uncertainty, i.e.,
\begin{equation}
H(Z) = -\sum p(z)\log p(z)  ,
\end{equation}
where $p(z)$ denotes the underlying probability mass function of the discrete random variable $Z$. Also, another important quantity to consider is the \textit{conditional entropy}. It is defined by 
\begin{equation}
H(Z|V) = -\sum p(z,v)\log p(z|v)  ,
\end{equation}
and represents the average uncertainty in $Z$ when a realization of the random variable $V$ is known.

The dynamical properties of our coupled system can then be studied by introducing the notion of the \textit{transition probability} which is the probability of transition of the discrete-time system from its past states to the current one. Denote by $Z_t$ the present state of the process (time series) $Z$, while $\mathbf{Z}_t^{-} = [Z_{t-1},Z_{t-2},\cdots]^\intercal$ represents the \textit{state variable} of the process $Z$, which contains the complete information related to the past of the process $Z$. The \textit{self-entropy} for $Z$ is then defined as follows:
\begin{equation}
\label{Eq:Self_Ent}
S_Z = \sum p(z_t,\mathbf{z}_t^{-})\log \dfrac{p(z_t|\mathbf{z}_t^{-})}{p(z_t)}\CommaPunct
\end{equation}
where $z_t$ and $\mathbf{z}_t^{-}$ denote the corresponding realizations of $Z_t$ and $\mathbf{Z}_t^{-}$, respectively. Basically, the self-entropy  assesses the influence of the past states of the target process $Z$ onto its present state. In other words, the self-entropy given by \cref{Eq:Self_Ent} measures the extent to which the uncertainty about $Z_t$ is reduced by the knowledge of $\mathbf{Z}_t^{-}$. The self-entropy for the process $Z$ ranges from zero to $H(Z_t)$, where the zero corresponds to the situation where the information from the past states ($\mathbf{Z}_t^{-}$) does not contribute to reduction of uncertainty in $Z_t$, and the latter quantity is associated to the case where the whole uncertainty is dissipated about $Z_t$ after acquiring the knowledge of $\mathbf{Z}_t^{-}$.

Besides, in order to assess the influence of imparting the past states of the process $V_t$  on $Z_t$, i.e. to evaluate the information contained in $\mathbf{V}_t^{-}$, one can make use of the well-known \textit{transfer entropy} measure \Citep{schreiber2000measuring}, which is given by the following formula:
\begin{equation}
\label{Eq:Transfer_Ent}
T_{V\to Z} = \sum p(z_t,\mathbf{v}_t^{-},\mathbf{z}_t^{-})\log \dfrac{p(z_t|\mathbf{v}_t^{-},\mathbf{z}_t^{-})}{p(z_t|\mathbf{z}_t^{-})}  \cdot
\end{equation}
Similar to self-entropy, the smallest value $T_{V\to Z}=0$ is attained when $\mathbf{V}_t^{-}$ does not induce an uncertainty reduction in $Z_t$ beyond what is already contributed by $\mathbf{Z}_t^{-}$. In contrast, the largest value of transfer entropy  $T_{V\to Z}=H(Z_t|\mathbf{Z}_t^{-})$ is taken on when the totality of uncertainty about $Z_t$ that was not already dissipated by the knowledge of $\mathbf{Z}_t^{-}$ is reduced by the information provided through imparting $\mathbf{V}_t^{-}$. In other words, the direct information measure, signified by the transfer entropy, quantifies the loss of information when assuming $V_t$ does not causally influence $Z_t$ when it actually does, e.g., see to \Citep{amblard2013relation} for more details.

Note that the self-entropy and transfer entropy measures introduced above, can also be expressed as follows:
\begin{equation}
S_Z = H(Z_t) - H(Z_t|\mathbf{Z}_t^{-})  ,
\end{equation}
and
\begin{equation}
T_{V\to Z} = H(Z_t|\mathbf{Z}_t^{-}) - H(Z_t|\mathbf{V}_t^{-},\mathbf{Z}_t^{-})  .
\end{equation}

The values of self-entropy and transfer entropy computed\footnote{Transfer entropy quantities were computed using the RTransferEntropy package in R programming language \Citep{behrendt2019rtransferentropy}, while the self-entropy quantities were computed using the ITS MATLAB toolbox for practical computation of information dynamics available online at \url{http://www.lucafaes.net/its.html} which is developed based on the methodologies presented in \Citep{faes2015information,faes2017information,xiong2017entropy}.} using daily realizations of $\mathrm{FCIX}_t$ and $\mathrm{VIX}_t$ during the time period from January 1990 to December 2019 are depicted in \Cref{fig:TE_Bivariate_general}, whereby it can be seen that $T_{\mathrm{VIX}\to \mathrm{FCIX}}$ is considerably larger than $T_{\mathrm{FCIX}\to \mathrm{VIX}}$. This suggests that the time-dependent information transfer is much stronger in one direction than the other one. This observation supports our previous discussions on orthogonalized impulse-response functions where it was shown that the magnitude of impact on $\mathrm{FCIX}_t$ responding to a shock in $\mathrm{VIX}_t$ was observed to be much greater than that of $\mathrm{VIX}_t$ reacting to shock in $\mathrm{FCIX}_t$. Therein, it was also discussed that a shock to $\mathrm{VIX}_t$ had a contemporaneous effect on $\mathrm{FCIX}_t$, but not vice versa. 

\begin{figure}
	\centering
	\begin{tikzpicture}[thick,scale=1, every node/.style={scale=0.7,color=black}, LabelStyle/.style = { rectangle, rounded corners, draw,minimum width = 2em, fill= yellow!50, text = black, font = \large\bfseries }, VertexStyle/.style ={draw, shape = circle,
		shading = ball, color=white, ball color = red!50!black, minimum size = 50pt, font = \Large\bfseries} ]
	\SetGraphUnit{4}
	\Vertex{FCIX}
	\EA[Lpos=90,unit=3](FCIX){VIX}
	\Edge[label = $0.005$](FCIX)(VIX)
	\Edge[label = $0.022$](VIX)(FCIX)
	\Loop[dist=2cm,dir=NO,style={thick,->},label = $0.678$](FCIX.west)
	\Loop[dist=2cm,dir=SO,style={thick,->},label = $1.671$](VIX.east)
	\end{tikzpicture}
	\caption[Computed diagram of information flow  between $\mathrm{FCIX}_t$ and $\mathrm{VIX}_t$.]{Computed diagram of information flow  between $\mathrm{FCIX}_t$ and $\mathrm{VIX}_t$ during January 1990-December 2019.}%
	\label{fig:TE_Bivariate_general}
\end{figure}
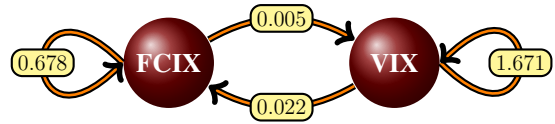

A subsequent combination of the results obtained using the above information-theoretic measures, namely, the transfer entropy and self-entropy, with those derived by employing the orthogonalized impulse-response functions leads us to draw a conclusion on  the direction of causal relations between these two processes. Namely, we conclude that there is a substantial evidence favoring existence of a bi-directional causal relation between $\mathrm{FCIX}_t$ and $\mathrm{VIX}_t$. Furthermore, this causal relation implies that $\mathrm{VIX}_t$ influences $\mathrm{FCIX}_t$ instantaneously with a significant dispense of information, whereas the opposite casual relation does not share a similar instantaneous nature. Yet, it still contributes to a substantial transfer of information from $\mathrm{FCIX}_t$ to $\mathrm{VIX}_t$, where the dispensed information are most likely to be transferred in a time-lagged manner. 

Next, we extend our analysis to the situation where one concerns with the evolution of information dynamics between $\mathrm{FCIX}_t$ and $\mathrm{VIX}_t$ over time. To this end, we formulate the following dynamical system model to describe the various mechanisms of exchange of information between the considered processes. Namely, we propose the following time-dependent dynamical system model:
\begin{equation}
\label{Eq:Dynamical_System}
\begin{cases}
\dfrac{\partial F}{\partial t} = (\gamma+\theta) F^2 - \alpha F  , \\ \\
\dfrac{\partial V}{\partial t} = \beta V^2 - \delta V   ,
\end{cases}
\end{equation}
where $F:=\mathrm{FCIX}_t$ and $V:=\mathrm{VIX}_t$. Also, parameters $\alpha$, $\beta$, $\gamma$, $\delta$ and $\theta$ represent $T_{\mathrm{FCIX}\to \mathrm{VIX}}$, $T_{\mathrm{VIX}\to \mathrm{FCIX}}$, $S_{\mathrm{FCIX}}$, $S_{\mathrm{VIX}}$, $\mathrm{ApEn(iVrF)}$, respectively, where $\mathrm{iVrF}$ denotes the orthogonalized response function of $\mathrm{FCIX}_t$ to an impulse in $\mathrm{VIX}_t$. These parameters are also illustrated in \Cref{fig:TE_Bivariate}. It deserves to mention that, the term $\theta$ is involved in model \eqref{Eq:Dynamical_System} as a result of the contemporaneous effect that a one-unit change in $\mathrm{VIX}_t$ exerts on $\mathrm{FCIX}_t$. Note that the approximate entropy can potentially be viewed as a measure of information content in the orthogonalized IRF for $\mathrm{VIX}_t \leadsto \mathrm{FCIX}_t $ during the specified time horizon.

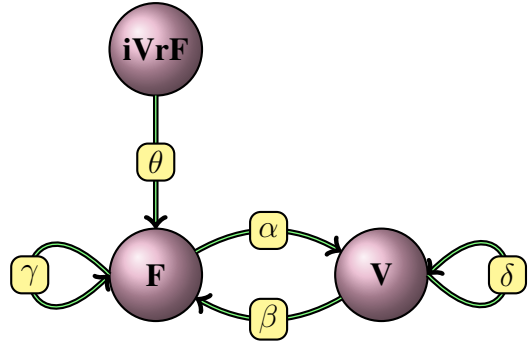
\begin{figure}
	\centering
	\begin{tikzpicture}[thick,scale=1, every node/.style={scale=0.7}, LabelStyle/.style = { rectangle, rounded corners, draw,minimum width = 2em, minimum height = 2em, fill= yellow!50, text = black, font = \LARGE\bfseries } , VertexStyle/.style ={draw, shape = circle,
		shading = ball, color=black, ball color = pink!90!blue, minimum size = 50pt, font = \LARGE\bfseries} , EdgeStyle/.style =
	{->,bend left, thick, double = white!50!green, double distance = 0.5pt} ]
	\SetGraphUnit{4}
	\Vertex{F}(FCIX)
	\Vertex[x=3,y=0]{V}(VIX)
	\Vertex[x=0,y=3]{iVrF}(IRF)
	\Edge[label = $\alpha$](FCIX)(VIX)
	\Edge[label = $\beta$](VIX)(FCIX)
	\Edge[label = $\theta$,style={bend left=0}](iVrF)(FCIX)
	\Loop[dist=2cm,dir=NO,style={thick,->},label = $\gamma$](FCIX.west)
	\Loop[dist=2cm,dir=SO,style={thick,->},label = $\delta$](VIX.east)
	\end{tikzpicture}
	\caption{Symbolic diagram of information flow  between $\mathrm{FCIX}_t$ and $\mathrm{VIX}_t$.}%
	\label{fig:TE_Bivariate}
\end{figure}

In order to characterize the model \eqref{Eq:Dynamical_System} in more in-depth, we perform its stability analysis. It can be shown with some effort that the four critical points of this dynamical system are as follows:
\begin{equation*}
(F_1,V_1) =  (\frac{\alpha}{\gamma+\theta},\frac{\delta}{\beta})  , \,\, 
(F_2,V_2) =  (0,0)  , \,\, 
(F_3,V_3) =  (\frac{\alpha}{\gamma+\theta},0) , \,\,  \text{and} \,\, 
(F_4,V_4) =  (0,\frac{\delta}{\beta})  \cdot
\end{equation*}
Note that these critical points are all feasible as $\mathrm{FCIX}_t$ and $\mathrm{VIX}_t$ can take on zero values from the theoretical point of view. However, such cases are not practical and we preclude them from our analysis. Also, without the loss of generality, we impose a set of non-zero constraints on $\beta$ and $ \gamma+\theta$. These constraints will prevent unboundedness of the elements of $(F_1,V_1)$ which can lead to the cases of unbounded realized as well as implied types of market volatility. 

Next, the Jacobian matrix is obtained as 
\begin{equation}
J =
\begin{pmatrix}
2(\gamma+\theta)F-\alpha & \theta \\
\theta &  2\beta V - \delta
\end{pmatrix}  .
\end{equation}
Subsequently, we compute the eigenvalues at $(F_1,V_1)$, which yield
\begin{equation*}
(\lambda_1,\lambda_1^{'}) =  (\frac{\alpha}{\delta+\theta},\frac{\delta}{\beta})  .
\end{equation*}

Thereafter, we apply the \textit{center manifold theorem} to study stability of the critical points by deriving their dimensions of stable ($\mathrm{Dim}(E^S)$), unstable ($\mathrm{Dim}(E^U)$) and center ($\mathrm{Dim}(E^C)$) manifolds. Using the previously computed values of $\alpha=0.005$, $\beta=0.022$, $\gamma=0.678$, $\delta=1.671$ and $\theta=0.160$, the first critical point is evaluated to be $(F_1,V_1)=(0.006,75.954)$. This critical point also has the following manifolds' dimensions:
\begin{equation*}
\mathrm{Dim}(E^S) = 0  , \, \mathrm{Dim}(E^U) = 2,  \text{ and }  \mathrm{Dim}(E^C) = 0  ,
\end{equation*}
which indicates that in fact it is a saddle point.

It is also interesting to note that the elements of $(F_1,V_1)=(0.006,75.954)$ reside on roughly all-time highs of their respective time series $\mathrm{FCIX}_t$ ($\max=0.0073$ and $\mathrm{mean}=0.0006$) and $\mathrm{VIX}_t$ ($\max=80.860$ and $\mathrm{mean}=19.146$). In other words, there could potentially exist hidden acting forces which prevent the market from retaining its volatility state at the neighborhoods of extrema. This assertion can be verified by noting that the average values for $\mathrm{FCIX}_t$ and $\mathrm{VIX}_t$ during the considered time period January $1990$-December $2019$ are both substantially less than their extrema, an observation which could well be due to the fact that $(F_1,V_1)$ is a source.

\section{Conclusions}
\label{Section:8}
In this paper, we first axiomitized the pairwise comparative judgments and demonstrated that the so-called reciprocal pairwise comparison matrices were the only members of the family of pairwise comparison matrices that would satisfy the formulated axiom. Subsequently, various algebraic properties of the reciprocal pairwise comparison matrices such as their rank, trace, powers, eigenvalues, etc, were discussed, and it was shown that the reciprocal pairwise comparison matrices are suitable structures to model the procedures that underlie agents' thought processes and their mechanisms of casting judgments. Thereafter, the sensitivity of the eigenvalues of the reciprocal pairwise comparison matrices to perturbations in their entries were investigated, and it was established that the largest eigenvalues of the considered matrices would be robust w.r.t. such perturbations.

Next, we defined a special class of spatio-temporal tensors by extending the reciprocal pairwise comparison matrices to a tensor domain, which in turn enabled us to effectively embed the collective judgment of the agents throughout time. An inconsistency function was then developed based on the largest eigenvalues of the frontal slices of the constructed tensors, in order to examine the consistency of the judgments cast by the agents throughout time. A relationship between the mentioned inconsistency function and the rank-$1$ estimates of the constructed tensors was further established, upon which the financial chaos index was defined. Several properties of the financial chaos index were also studied, including its robustness for estimating the market volatility, which enabled us to reliably perform the tasks of market segmentation by resorting to the applications of the proposed index.

Furthermore, we clarified the connection between the equity and option markets by studying the relationship between the financial chaos index (as a measure of market's realized volatility) and the VIX (as a measure of market's implied volatility). More so, the causal relations existing among the financial chaos index and VIX were exploited using the orthogonalized impulse response functions and information-theoretic measures. Our computational results which pertain to the time period January 1990- December 2019 imply that there exist a bidirectional causal relation between the processes underlying the realized and implied volatility of the stock market within the given time period, where it was shown that the later had a stronger causal effect on the former as compared to the opposite.
%Subsequently, the results obtained by analyzing the causal relations mentioned above, were employed in testing the efficient market hypothesis. It was shown that the U.S. stock market for the time frame January $1990$-December $2019$, for either of the monthly or quarterly time frequencies, was weakly, potently and semi-strongly efficient at $1\%$ level of significance. On the contrary, the stock market during this time period turned out to be weakly, potently and semi-strongly inefficient as per daily time frequency at all three levels of significance, namely, $1\%$, $5\%$ and $10\%$.

%%%%%%%%%%%%%%%%%%%%%%%%%%%%%%%%%%%%%%%%%%%%%%%%%%%%%%%%%%%%%%%%%%%%%%%%%%%%%%%%%%%%%%%%%%%%%%%%%
%%%%%%%%%%%%%%%%%%%%%%%%%%%%%%%%%%%%%%%%%%%%%%%%%%%%%%%%%%%%%%%%%%%%%%%%%%%%%%%%%%%%%%%%%%%%%%%%%
%%%%%%%%%%%%%%%%%%%%%%%%%%%%%%%%%%%%%%%%%%%%%%%%%%%%%%%%%%%%%%%%%%%%%%%%%%%%%%%%%%%%%%%%%%%%%%%%% 

\section*{Acknowledgments}
This work was supported by Natural Sciences and Engineering Research Council of Canada (NSERC). The authors would like to gratefully acknowledge contributions of Professor Vladimir Vinogradov for his fruitful discussions regarding the statistical procedures presented in this paper. The authors are further grateful to Professor Torben G. Anderson for providing constructive comments and feedback on the employed statistical procedures as well as economical interpretation of the reported results.

%%%%%%%%%%%%%%%%%%%%%%%%%%%%%%%%%%%%%%%%%%%%%%%%%%%%%%%%%%%%%%%%%%%%%%%%%%%%%%%%%%%%%%%%%%%%%%%%%
%%%%%%%%%%%%%%%%%%%%%%%%%%%%%%%%%%%%%%%%%%%%%%%%%%%%%%%%%%%%%%%%%%%%%%%%%%%%%%%%%%%%%%%%%%%%%%%%%
%%%%%%%%%%%%%%%%%%%%%%%%%%%%%%%%%%%%%%%%%%%%%%%%%%%%%%%%%%%%%%%%%%%%%%%%%%%%%%%%%%%%%%%%%%%%%%%%%   

\medskip

\printbibliography

@book{suppes1999introduction,
	title={Introduction to Logic},
	author={Suppes, Patrick},
	year={1999},
	publisher={Dover Publications}
}

@article{fishburn1970intransitive,
	title={Intransitive individual indifference and transitive majorities.},
	author={Fishburn, Peter C},
	journal={Econometrica: Journal of the Econometric Society,},
	volume = {38},
	number = {3},
	pages={482--489},
	year={1970},
	publisher={JSTOR}
}

@article{saaty1977scaling,
	title={A scaling method for priorities in hierarchical structures.},
	author={Saaty, Thomas L},
	journal={Journal of Mathematical Psychology,},
	volume={15},
	number={3},
	pages={234--281},
	year={1977},
	publisher={Elsevier}
}

@book{harville2018linear,
	title={Linear Models and the Relevant Distributions and Matrix Algebra},
	author={Harville, David A},
	year={2018},
	publisher={CRC Press}
}

@article{delgado2019approximate,
	title={Approximate entropy and sample entropy: A comprehensive tutorial.},
	author={Delgado-Bonal, Alfonso and Marshak, Alexander},
	journal={Entropy,},
	volume={21},
	number={6},
	pages={541},
	year={2019},
	publisher={Multidisciplinary Digital Publishing Institute}
}

@article{gibbard2014intransitive,
	title={Intransitive social indifference and the Arrow dilemma.},
	author={Gibbard, Allan F},
	journal={Review of Economic Design,},
	volume={18},
	number={1},
	pages={3--10},
	year={2014},
	publisher={Springer}
}

@book{arrow2012social,
	title={Social Choice and Individual Values},
	author={Arrow, Kenneth J},
	volume={12},
	year={2012},
	publisher={Yale university press}
}

@article{truong2018review,
	title={A review of change point detection methods.},
	author={Truong, Charles and Oudre, Laurent and Vayatis, Nicolas},
	journal={ArXiv:1801.00718,},
	pages={1--31},
	year={2018}
}

@article{arlot2019kernel,
	title={A kernel multiple change-point algorithm via model selection.},
	author={Arlot, Sylvain and Celisse, Alain and Harchaoui, Zaid},
	journal={Journal of Machine Learning Research,},
	volume={20},
	number={162},
	pages={1--56},
	year={2019}
}

@article{desobry2005online,
	title={An online kernel change detection algorithm.},
	author={Desobry, Fr{\'e}d{\'e}ric and Davy, Manuel and Doncarli, Christian},
	journal={IEEE Transactions on Signal Processing,},
	volume={53},
	number={8},
	pages={2961--2974},
	year={2005},
	publisher={IEEE}
}

@inproceedings{harchaoui2007retrospective,
	title={Retrospective mutiple change-point estimation with kernels.},
	author={Harchaoui, Zaid and Capp{\'e}, Olivier},
	booktitle={2007 IEEE/SP 14th Workshop on Statistical Signal Processing},
	pages={768--772},
	year={2007},
	organization={IEEE}
}

@inproceedings{harchaoui2009kernel,
	title={Kernel change-point analysis.},
	author={Harchaoui, Zaid and Moulines, Eric and Bach, Francis R},
	booktitle={Advances in Neural Information Processing Systems},
	pages={609--616},
	year={2009}
}

@inproceedings{sriperumbudur2008injective,
	title={Injective Hilbert space embeddings of probability measures.},
	author={Sriperumbudur, Bharath K and Gretton, Arthur and Fukumizu, Kenji and Lanckriet, Gert and Sch{\"o}lkopf, Bernhard},
	booktitle={21st Annual Conference on Learning Theory (COLT 2008)},
	pages={111--122},
	year={2008},
	organization={Omnipress}
}

@book{cover2012elements,
	title={Elements of Information Theory},
	author={Cover, Thomas M and Thomas, Joy A},
	year={2012},
	publisher={John Wiley \& Sons}
}

@Inbook{faes2014conditional,
	author={Faes, Luca and Porta, Alberto},
	title={Conditional entropy-based evaluation of information dynamics in physiological systems.},
	bookTitle={Directed Information Measures in Neuroscience},
	year={2014},
	publisher={Springer},
	pages={61--86}
}

@article{schreiber2000measuring,
	title={Measuring information transfer.},
	author={Schreiber, Thomas},
	journal={Physical Review Letters,},
	volume={85},
	number={2},
	pages={461},
	year={2000},
	publisher={APS}
}

@book{wilkinson1965algebraic,
	title={The Algebraic Eigenvalue Problem},
	author={Wilkinson, James H},
	year={1988},
	publisher = {Oxford University Press, Inc.},
	address = {USA}
}

@article{smith1967condition,
	title={The condition numbers of the matrix eigenvalue problem.},
	author={Smith, Russell A},
	journal={Numerische Mathematik,},
	volume={10},
	number={3},
	pages={232--240},
	year={1967},
	publisher={Springer}
}

@article{saaty1993relative,
	title={What is relative measurement? The ratio scale phantom.},
	author={Saaty, Thomas L},
	journal={Mathematical and Computer Modelling,},
	volume={17},
	number={4-5},
	pages={1--12},
	year={1993},
	publisher={Elsevier}
}

@article{amblard2013relation,
	title={The relation between Granger causality and directed information theory: A review.},
	author={Amblard, Pierre-Olivier and Michel, Olivier JJ},
	journal={Entropy,},
	volume={15},
	number={1},
	pages={113--143},
	year={2013},
	publisher={Multidisciplinary Digital Publishing Institute}
}

@article{pincus2008approximate,
	title={Approximate entropy as an irregularity measure for financial data.},
	author={Pincus, Steve},
	journal={Econometric Reviews,},
	volume={27},
	number={4-6},
	pages={329--362},
	year={2008},
	publisher={Taylor \& Francis}
}

@article{behrendt2019rtransferentropy,
	title={RTransferEntropy—Quantifying information flow between different time series using effective transfer entropy.},
	author={Behrendt, Simon and Dimpfl, Thomas and Peter, Franziska J and Zimmermann, David J},
	journal={SoftwareX,},
	volume={10},
	pages={100265},
	year={2019},
	publisher={Elsevier}
}

@article{faes2015information,
	title={Information decomposition in bivariate systems: theory and application to cardiorespiratory dynamics.},
	author={Faes, Luca and Porta, Alberto and Nollo, Giandomenico},
	journal={Entropy,},
	volume={17},
	number={1},
	pages={277--303},
	year={2015},
	publisher={Multidisciplinary Digital Publishing Institute}
}

@article{faes2017information,
	title={Information decomposition in multivariate systems: definitions, implementation and application to cardiovascular networks.},
	author={Faes, Luca and Porta, Alberto and Nollo, Giandomenico and Javorka, Michal},
	journal={Entropy,},
	volume={19},
	number={1},
	pages={5},
	year={2017},
	publisher={Multidisciplinary Digital Publishing Institute}
}

@article{xiong2017entropy,
	title={Entropy measures, entropy estimators, and their performance in quantifying complex dynamics: Effects of artifacts, nonstationarity, and long-range correlations.},
	author={Xiong, Wanting and Faes, Luca and Ivanov, Plamen C},
	journal={Physical Review E,},
	volume={95},
	number={6},
	pages={062114},
	year={2017},
	publisher={APS}
}

@article{pincus1994physiological,
	title={Physiological time-series analysis: what does regularity quantify?.},
	author={Pincus, Steven M and Goldberger, Ary L},
	journal={American Journal of Physiology-Heart and Circulatory Physiology,},
	volume={266},
	number={4},
	pages={H1643--H1656},
	year={1994},
	publisher={American Physiological Society Bethesda, MD}
}

@article{pincus1991approximate,
	title={Approximate entropy as a measure of system complexity.},
	author={Pincus, Steven M},
	journal={Proceedings of the National Academy of Sciences,},
	volume={88},
	number={6},
	pages={2297--2301},
	year={1991},
	publisher={National Acad Sciences}
}

@article{ishizaka2009analytic,
	title={Analytic hierarchy process and expert choice: Benefits and limitations.},
	author={Ishizaka, Alessio and Labib, Ashraf},
	journal={OR Insight,},
	volume={22},
	number={4},
	pages={201--220},
	year={2009},
	publisher={Taylor \& Francis}
}

@article{ishizaka2011review,
	title={Review of the main developments in the analytic hierarchy process.},
	author={Ishizaka, Alessio and Labib, Ashraf},
	journal={Expert Systems With Applications.},
	volume={38},
	number={11},
	pages={14336--14345},
	year={2011},
	publisher={Elsevier}
}

@book{brunelli2014introduction,
	title={Introduction to the Analytic Hierarchy Process.},
	author={Brunelli, Matteo},
	year={2014},
	publisher={Springer}
}

@book{saaty2012models,
	title={Models, Methods, Concepts \& Applications of the Analytic Hierarchy Process.},
	author={Saaty, Thomas L and Vargas, Luis G},
	year={2012},
	publisher={Springer Science \& Business Media}
}

@incollection{mu2017understanding,
	title={Understanding the analytic hierarchy process.},
	author={Mu, Enrique and Pereyra-Rojas, Milagros},
	booktitle={Practical Decision Making},
	pages={7--22},
	year={2017},
	publisher={Springer}
}

@book{mu2016practical,
	title={Practical Decision Making: An Introduction to the Analytic Hierarchy Process (AHP) Using Super Decisions V2.},
	author={Mu, Enrique and Pereyra-Rojas, Milagros},
	year={2016},
	publisher={Springer}
}

@book{mu2017practical,
	title={Practical Decision Making Using Super Decisions V3: An Introduction to the Analytic Hierarchy Process.},
	author={Mu, Enrique and Pereyra-Rojas, Milagros},
	year={2017},
	publisher={Springer}
}

@article{wang1974permanents,
	title={On permanents of (1,- 1)-matrices.},
	author={Wang, Edward Tzu-Hsia},
	journal={Israel Journal of Mathematics},
	volume={18},
	number={4},
	pages={353--361},
	year={1974},
	publisher={Springer}
}

@article{akbari2016permanents,
	title={Permanents of matrices over roots of unity.},
	author={Akbari, S and Ariannejad, M and Tajfirouz, Z},
	journal={Linear and Multilinear Algebra},
	volume={64},
	number={9},
	pages={1769--1775},
	year={2016},
	publisher={Taylor \& Francis}
}

@article{johansen2014role,
	title={The role of initial values in conditional sum-of-squares estimation of nonstationary fractional time series models.},
	author={Johansen, S{\o}ren and Nielsen, Morten},
	journal={QED working paper 1300, Queen's University},
	year={2014}
}

@article{jensen2014fast,
	title={A fast fractional difference algorithm.},
	author={Jensen, Andreas N and Nielsen, Morten},
	journal={Journal of Time Series Analysis},
	volume={35},
	number={5},
	pages={428--436},
	year={2014},
	publisher={Wiley Online Library}
}

@article{abadir2007nonstationarity,
	title={Nonstationarity-extended local Whittle estimation.},
	author={Abadir, Karim M and Distaso, Walter and Giraitis, Liudas},
	journal={Journal of econometrics},
	volume={141},
	number={2},
	pages={1353--1384},
	year={2007},
	publisher={Elsevier}
}

@article{johansen2019nonstationary,
	title={Nonstationary cointegration in the fractionally cointegrated VAR model.},
	author={Johansen, S{\o}ren and Nielsen, Morten},
	journal={Journal of Time Series Analysis},
	volume={40},
	number={4},
	pages={519--543},
	year={2019},
	publisher={Wiley Online Library}
}

@article{johansen2008representation,
	title={A representation theory for a class of vector autoregressive models for fractional processes.},
	author={Johansen, S{\o}ren},
	journal={Econometric Theory},
	pages={651--676},
	year={2008},
	publisher={JSTOR}
}

@article{johansen2012likelihood,
	title={Likelihood inference for a fractionally cointegrated vector autoregressive model.},
	author={Johansen, S{\o}ren and Nielsen, Morten},
	journal={Econometrica},
	volume={80},
	number={6},
	pages={2667--2732},
	year={2012},
	publisher={Wiley Online Library}
}

@article{johansen2018testing,
	title={Testing the CVAR in the fractional CVAR model.},
	author={Johansen, S{\o}ren and Nielsen, Morten},
	journal={Journal of Time Series Analysis},
	volume={39},
	number={6},
	pages={836--849},
	year={2018},
	publisher={Wiley Online Library}
}

@article{dolatabadi2016fractionally,
	title={A fractionally cointegrated VAR model with deterministic trends and application to commodity futures markets.},
	author={Dolatabadi, Sepideh and Nielsen, Morten and Xu, Ke},
	journal={Journal of Empirical Finance},
	volume={38},
	pages={623--639},
	year={2016},
	publisher={Elsevier}
}

@article{nielsen2016matlab,
	title={A matlab program and user's guide for the fractionally cointegrated VAR model.},
	author={Nielsen, Morten and Popiel, Micha{\l} Ksawery},
	year={2016},
	journal={Queen's Economics Department Working Paper}
}

%
%\bigskip
%\footnotesize
%\noindent\textit{Acknowledgments.}
%This research was partly supported by NSF (grant no. XXXX).

%\bibliographystyle{tfs}
%\bibliography{references}

\end{document}